\documentclass[letter]{article}
\pdfoutput=1

\usepackage{fullpage}

\usepackage{graphics}
\usepackage{graphicx}
\usepackage{epsfig}
\usepackage{amsmath}
\usepackage{amssymb}
\usepackage{bbm}
\usepackage{mathtools}
\usepackage{enumitem}
\usepackage{color} 
\usepackage{float} 
\usepackage{caption}
\usepackage[toc,page]{appendix}

\usepackage[latin1]{inputenc} 
\usepackage[T1]{fontenc}

\def\R{\mathbb{R}}
\def\E{\mathbb{E}}
\def\N{\mathbb{N}}

\def\qed{\vspace{2mm} \vrule height4pt width3pt depth2pt}

\newcommand{\prim}{$'$}

\newtheorem{definition}{Definition} 
\newtheorem{proposition}{Proposition}

\newtheorem{remark}{Remark}

\title{Dynamics and biological reference points in a stochastic age-structured fish population model with an illustration of the Patagonian toothfish population}
\author{{V. Riquelme$^{1}$\footnote{This work was achieved while the first author was at Departamento de Ingeniería Matemática and Centro de Modelamiento Matemático (UMI CNRS 2807), Universidad de Chile, Beauchef 851, Santiago, Chile}, T.J. Quinn II$^2$, H. Ram\'irez C.$^{3}$}\\[2mm]
$^{1}$ Departamento de Matem\'atica, Universidad T\'ecnica Federico Santa Mar\'ia,\\
Avenida Espa\~na 1680, Valpara\'iso, Chile.\\[2mm]
{\tt vriquelme@dim.uchile.cl}\\[2mm]
$^2$ Juneau Center, College of Fisheries and Ocean Sciences, \\
University of Alaska Fairbanks, 11120 Glacier Highway, Juneau, AK 99801-8677, US\\
{\tt terry.quinn@alaska.edu}\\[2mm]
$^{3}$ Departamento de Ingenier\'ia Matem\'atica \& Centro de Modelamiento
Matem\'atico\\ (UMI 2807, CNRS), Universidad de Chile, Beauchef 851, Casilla 170-3,
Santiago 3, Chile.\\[2mm]
{\tt hramirez@dim.uchile.cl}\\[2mm]
}

\date{\today}

\begin{document}

\maketitle

{\em Abstract.}

In this manuscript we investigate the long-term behavior of a single-species fishery, which is harvested by several fleets. The time evolution of this population is modeled by a discrete time stochastic age-structured model. We assume that incertitude only affects the recruitment. First, for the deterministic version of this model, we characterize the equilibrium yields in terms of the fishing mortality. Then, for the stochastic version, we introduce the concepts of maximum expected, log expected and harmonic expected sustainable yield, and we analyze how the incertitude affects the behavior of these yields and their stationary distribution. All the numerical simulations are performed with data obtained from Patagonian Tooth-fish fishery, which is harvested by four different type of fleets: Chilean Industrial fleet, Chilean Artisanal fleet, Argentinean longline fleet, and Argentinean Artisanal fleet.


\section{Introduction}

Age-structured fish population dynamics models are ubiquitous around the world for integrating the diverse sources of information available with key parameters for describing the factors affecting the dynamics, including natural and fishing mortality and reproduction, along with the fishing process \cite{quinn}. The basis of these models is a deterministic, linear Leslie matrix formulation \cite{caswell,getzhaight,horwood,leslie}. For biological realism, two modifications are useful: (1) nonlinear dynamics during the early life history stage to obtain equilibrium or sustainability, and (2) stochastic variation in early life survival to account for the large amount of uncertainty due to environmental conditions. 

Stochastic variability in fish population models due to recruitment has already been studied in \cite{brodziak,getz,getzfrancis,hightower,horwood,reed2}. \cite{horwood} study the sensitivity of a age-structured model with respect to noise in general form, by means of Fourier analysis. In \cite{reed2} a discrete time non-linear stochastic age-structured population model without plus group is studied, and equations for approximations of the first and second moments of each age-group are obtained, in order to obtain the steady state variances of the recruitment and yield. A similar analysis is carried out in \cite{getz}, with two main differences: first, the model considers a plus group; secondly, each year is splitted into two seasons: during \emph{harvesting season} the model is given by an ordinary differential equation (ODE), and during \emph{spawning season} the model is given by a discrete time equation. Approximate and explicit expressions for the first and second moments of each age-group are obtained. A more practical approach is carried out in \cite{getzfrancis}, where the authors perform Monte-Carlo simulations of a stochastic age-structured model with a Ricker spawner-recruitment function to estimate the long term mean yield, for three different fisheries, and conclude that the maximum expected sustainable yield in all the studied fisheries decreases as the coefficient of variation ($CV$) of recruitment increases (considering a range for $CV$ from 0 to 200\%.). \cite{brodziak} considers a model similar to \cite{reed2}, with the addition of a plus group, Baranov catches, and a Beverton-Holt spawner-recruitment function with multiplicative log-normal noise (without correction for bias), and performs Monte-Carlo simulations of the obtained model in order to generate short-term stochastic projections for different constant values of fishing mortalities given by biological reference points. Nevertheless, explicit expressions for the equilibrium distribution are not developed in any of the previously cited works.

Regarding risk measures, \cite{thompson} studies the effects of stochasticity in the optimization of harvesting control rules, and shows the convenience of considering the arithmetic, geometric, and harmonic mean yields as indicators of the state of the catch. According to \cite{thompson}, these quantities are related to the attitude towards risk from the point of view of an agent that needs to design fishing control policies. For example, a limit control rule might be defined by the decision-theoretic optimum derived under a risk-neutral stance, while a target control rule might be defined by the decision-theoretic optimum derived under a risk-averse stance. A simple way to characterize this difference is as follows: the risk-neutral solution maximizes the expectation of stationary yield, while the risk-averse solution maximizes the expectation of log-sustainable yield. Harmonic mean is also used as a part of a precautionary control value for biomass-based control rules used by the U.S. North Pacific Fishery Management Council \cite{NPFMC}. Due to the inherent nonlinearity of the models, it is not easy to compute these quantities, nor to obtain the stationary distribution associated to these models. This can be done in one-dimensional biomass models, such as in \cite{busquet} and \cite{ewald}, where the authors study respectively a discrete- and continuous-time Schaefer population model with multiplicative noise for the biomass, for specific probability distributions of the noise. The authors are able to derive explicit formulas for the stationary distribution, and prove that the expected sustainable yield decreases as the variance of the noise increases. The drawback of this formulation is that the information about the age-structure dissapears; the discussion in \cite{tahvonen} ilustrates the interest of considering an age-structured approach over a biomass approach.
\smallskip 

In this article we are interested in the study of the long term behavior of a single-species fishery harvested by several fleets and subject to environmental randomness that affects the recruitment. More specifically, we want to derive optimal constant fishing strategies that maximize the expected long term yield. To this purpose, we model the time evolution of the fishery by a discrete time non-linear stochastic age-structured population model with a plus group, where the recruitment is given by a Beverton-Holt spawner-recruitment function. We introduce the concepts of maximum expected stationary yield (MESY), maximum expected log-sustainable yield (MELSY), and maximum expected harmonic sustainable yield (MEHSY), and illustrate the case of the Patagonian toothfish population in Chilean and Argentinean fisheries, comparing the different reference points with the deterministic maximum sustainable yield (MSY) and among them. Via numerical simulations we show that the coefficient of variation has a negative effect on any of the maximum expected reference points, which advises to be more cautious if large levels of variability on recruitment are present in the fishery, and that the deterministic MSY cannot be attained in presence of environmental noise. These results confirm the theoretical results obtained in \cite{busquet} for biomass-based models.
\smallskip 

This article is organized as follows: In Section \ref{sec_def} we define the deterministic age-structured model and investigate properties of its equilibrium values depending on fishing mortality. In Section \ref{sec:stochasticModel} we introduce a stochastic term to the recruitment function, leading to a stochastic age-structured model, and we define the mentioned expected yield measures. In Section \ref{sec:numSim} we show the results of numerical simulations for both the deterministic and stochastic model, and we give estimates of these yield measures, showing the relevance of accounting for variability in the model.

\section{Materials and methods}
\subsection{Preliminaries and Deterministic Model}
\label{sec_def}

Consider a model of a fishery comprised of $n$ fleets $f=\{1,\dots,n\}$. The fish population contains individuals of ages $a\in\{1\,\dots,A-1,A\}$, in which $A$ is the age category of individuals of $A$ or more years, called a plus group. The age at which individuals reach $50\%$ maturity is at $a=r$, and the proportion of mature individuals at age $a$ is denoted $m_a$. 
%
%
The population within each age group is subject to a mortality that is composed by a natural mortality and the effect of fishing exerted by the fleets. We denote by $M_a$ the natural mortality rate by age $a$, and by $F_a$ the full recruitment fishing mortality rate by age $a$. This mortality rate is a linear function of the fishing effort $F$, that is, $F_a=\sum_{f=1}^nP_fs_{f,a}F$, where $P_f$ is the proportion of fishing mortality exerted by fleet $f$ and $s_{f,a}$ is the selectivity of fish of age $a$ by fleet $f$. Note that, by convention, $\max_a s_{f,a}=1$.
%
%
\begin{center}
\captionof{table}{Notation used in the age-structured model}
\begin{tabular}{||c|l||}\hline\hline
   $f$     & fleet type, $f\in\{1,\dots,n\}$ \\
   $a$     & age category, $a\in\{1,\dots,A-1,A\}$\\ 
   $M_a$   & natural mortality by age \\ 
   $P_f$   & proportion of fishing mortality of fleet $f$\\
   $F$     & full recruitment fishing mortality \\
$s_{f,a}$  & selectivity of fish of age $a$ by fleet $f$ \\
$F_a = F_a(F) = \sum_{f=1}^nP_fs_{f,a}F$ & fishing mortality at age $a$ \\
$Z_a= Z_a(F) = M_a + F_a(F)$ & total mortality rate for age $a$ \\
${\rm SSB}_t$ & spawning stock biomass at time $t$\\
$\tilde R_t$ & recruitment (at age $a=r$) at time $t$ \\
   $W_a$   & weight at age $a$ \\
   $m_a$   & maturity at age $a$ \\
   $R_0$   & virginal recruitment \\ \hline\hline
\end{tabular}
\end{center}

Total abundance $N$ is considered as the sum of all individuals of ages $r$ or more. The spawning stock biomass (SSB) at the year $t$ is composed by the individuals from the beginning of the year $t$ that have survived the proportion $\tau$ of the year until they spawn. 
\begin{equation}\label{eq:SSB}
{\rm SSB}_{t} = \sum_{a=r}^{A} m_aW_aN_{t,a}e^{-\tau Z_a}.
\end{equation}

Recruitment to the population occurs at age $r$ and is often the starting age at which data are collected. It depends on the SSB of $r$ years before. 
\begin{equation}\label{eq:Recruitment}
\tilde R_{t} \,=\, \varphi({\rm SSB}_{t-r}).
\end{equation}
The function $\varphi(\cdot)$ represents a spawner-recruit function, of which the asymptotic Beverton-Holt and the dome-shape Ricker are the most common. The Beverton-Holt equation is given by
\begin{equation}\label{eq:BHfunction}
\varphi({\rm SSB}) = \frac{\alpha {\rm SSB}}{\beta + {\rm SSB}}.
\end{equation}
(see, for instance, \cite{quinn}), where $\alpha = 4hR_0/(5h-1)$ and $\beta = B_0(1-h)/(5h-1)$ is the virginal biomass (unfished) and $h$ is the steepness of stock-recruit relationship. The individuals of age $a+1$ at a given year $t+1$ are the individuals of age $a$ that survived the year $t$. With the previous assumptions, the equations for the age-structured deterministic population model are
\begin{equation}\label{eq:modelBacalao}
\begin{split}
N_{t+1,1} = \varphi({\rm SSB}_t),\quad N_{t+1,a+1} &= N_{t,a}e^{-Z_a},\,\, a=1,\dots,A-2,\quad
N_{t+1,A} = N_{t,A-1}e^{-Z_{A-1}}+N_{t,A}e^{-Z_{A}}.
\end{split}
\end{equation}
The previous equations can be put in matrix form as
\begin{equation}\label{eq:matrix_form}
N_{t+1} = A(F)N_t + B\varphi({\rm SSB}_{t}),
\end{equation}
where $A(F)$ is an age-structured Leslie matrix in which the first row has null entries and the off-diagonals contain survival rates $\exp(-Z_a)$. In order to use this matrix formulation we place recruitment at $r$ years after their spawning; following \cite[Section~7.4]{quinn} we consider the natural mortalities of the first $r-1$ groups of juvenile individuals are null, this is, $M_1=\dots=M_{r-1}=0$, and the selectivities $s_{f,1}=\dots=s_{f,r-1}=0$ for all fleets $f=1,\dots,n$. To handle the plus group, it is further modified to take into account the survival rate of the individuals that belong to the age category $A$, or, 
\begin{equation*}
A(F)_{a+1,a} = e^{-Z_a(F)},\quad a=1,\dots,A-1;\qquad A(F)_{A,A} = e^{-Z_{A}(F)},
\end{equation*}
and $B=(1,0,\dots,0)^{\top}$. The yield is given by the Baranov catch in weight equation:
\begin{equation}\label{eq:yield}
Y_t = \sum_{a=r}^{A} W_a \frac{F_a}{Z_a}(1-e^{-Z_a})N_{t,a}.
\end{equation}

A key role in the analysis of population dynamics is the concept of equilibrium. For the standard Leslie matrix without including the plus group, these formulas were presented in \cite[Section~7.4]{quinn}. For this analysis, we consider constant mortality rates $M_a$, selectivities $s_{f,a}$, proportions of fishing mortalities $P_f$, and fishing mortality $F$. Including the plus group, the equilibrium values $N^{\star}=(N_1^{\star},\dots,N_{A}^{\star})^{\top}$ and ${\rm SSB}^{\star}$  can be computed in terms of $N^{\star}_1$ and the spawning potential ratio ${\rm SPR}^{\star}={\rm SPR}^{\star}(F)$, defined as
\begin{equation*}
{\rm SPR}^{\star} := \sum_{a=r}^{A-1}W_am_a\mathcal L_ae^{-\tau Z_a} + \frac{W_{A}m_{A}\mathcal L_{A}}{1-e^{-Z_{A}}}e^{-\tau Z_{A}},
\end{equation*}
where  $\mathcal L_a=\mathcal L_a(F)$ is the cumulative survival at age $a$ given by
\begin{equation*}
\mathcal L_a :=\exp\left\{-\sum_{x=1}^{a-1}Z_x\right\},\quad a=2,\dots,A,\quad\mathcal L_1:=1;\\
\end{equation*}

As in \cite{getz}, the population at equilibrium can be characterized by the equations
\begin{equation}\label{eq:equilib_N}
N_1^{\star} = \varphi(N_1^{\star}{\rm SPR}^{\star}),\qquad
N_{a}^{\star} = N_1^{\star}\mathcal L_a,\,\, a=2,\dots,A-1,\qquad
N_{A}^{\star} = N_{1}^{\star}\frac{\mathcal L_{A}}{1-e^{-Z_{A}}}.
\end{equation}
(a detailed deduction of previous formulas is stated in the Appendix \ref{app:1}). Notice that the spawning stock biomass ${\rm SSB}^{\star}$ and the abundances at equilibrium $N^{\star}$ depend on $N_1^{\star}$ as well as on the fishing mortality $F$. For the particular case of the Beverton-Holt spawner recruit function given in \eqref{eq:BHfunction}, we have 
\begin{equation}\label{eq:equilib_N1}
N_1^{\star} = \alpha-\frac{\beta}{{\rm SPR}^{\star}},\quad {\rm SSB}^{\star} = \alpha {\rm SPR}^{\star}-\beta.
\end{equation}
The yield at equilibrium has an explicit formula, given by
\begin{equation}\label{eq:equilib_Y}
Y^{\star} = Y^{\star}(F) = N_1^{\star}\left(\sum_{a=r}^{A-1} W_a \frac{F_a}{Z_a}(1-e^{-Z_a})\mathcal L_{a} + W_{A} \frac{F_{A}}{Z_{A}}\mathcal L_{A} \right).
\end{equation}

A condition for the equilibrium point to have ecological meaning is that ${\rm SPR^{\star}}\geq \frac{\beta}{\alpha}$. Also, \cite{reed} showed that a sufficient condition for stability of the equilibrium point is 
\begin{equation}
\left| {\rm SPR}^*\frac{\partial \varphi({\rm SSB}^{\star})}{\partial {\rm SSB}} \right|<1.
\end{equation}
(see \cite{quinn}). For the particular case of a Beverton-Holt spawner-recruit function of the form given in \eqref{eq:BHfunction}, the previous condition is exactly equivalent to ${\rm SPR^{\star}}> \frac{\beta}{\alpha}$, this is, the condition of positivity of the equilibrium ${\rm SSB^{\star}}$ (or, equivalently, positivity of $N_1^{\star}$).

\smallskip

\begin{definition}
We define the maximum sustainable yield $Y^{\star}_{\rm MSY}$ as the maximum yield at equilibrium (as a function of the fishing mortality $F$):
\begin{equation}
Y^{\star}_{\rm MSY} := \max_{F\geq0}\, Y^{\star}(F),
\end{equation}
where $Y^{\star}(F)$ is given in \eqref{eq:equilib_Y}. We denote $F_{\rm MSY}$ as the fishing mortality that produces the maximum sustainable yield, that is, $Y^{\star}_{\rm MSY}=Y^{\star}(F_{\rm MSY})$.
\end{definition}

According to Proposition \ref{prop:existencia} in Appendix \ref{app:1}, if ${\rm SPR^{\star}}(F=0)>\frac{\beta}{\alpha}$, then there exists a value $F_{\rm ext}>0$ such that for every $F>F_{\rm ext}$ the theoretical equilibrium $Y^{\star}(F)$ is non-positive, meaning that the fishery reaches its extinction. Then, the fishing mortality $F_{\rm MSY}$ that achieves the maximum sustainable yield $Y^{\star}_{\rm MSY}$ belongs to the interval $[0, F_{\rm ext}]$.


\subsection{Stochastic Model and Optimal Yield Measures}
\label{sec:stochasticModel}


Fluctuations on fish populations naturally appear as effect of environmental variations such as temperature, food availability, or reproductive success \cite{noise}. As in \cite{getz,hightower,reed}, we consider that the variability enters at the level of the stock-recruitment relationship by introducing a modification in the recruitment function given by a log-normal random variable at each time (see \cite{hightower} for the discussion about this particular choice of noise). The generalized model is
\begin{equation}\label{eq:stochModel}
N_{t+1} = A(F)N_t + B\varphi({\rm SSB}_{t})\omega_{t}e^{-\frac{1}{2}CV^2},
\end{equation}
where $(\omega_{t})_{t\in\N}$ is a sequence of independent and identically distributed lognormal random variables with parameters $\mu=0$ and $\sigma=CV$ (the coefficient of variation of the lognormal distribution) and independent of $N_0$. The term $\exp(-1/2CV^2)$ corrects for bias.

For each time $t\geq0$ and constant fishing effort $F\geq0$, the yield $Y_t(F)$ (given by \eqref{eq:yield}) is a random variable. As in the previous section, we aim to study the concept of maximum sustainable yield, but since yield is now a random variable it is necessary to analyze its probability distribution, and some notions of mean value. 
From now on, denote $Y^{\star}_F$ as the stationary yield, that is, the yield under its stationary behavior, whose probability distribution (that we denote $f_{Y^{\star}_F}(\cdot|F)$) depends on the mortality $F$ (that we assume to be constant). Following \cite{thompson}, we introduce the following definitions:

\if{
\begin{definition}
Define $Y^{\star}_{\rm ESY}=Y^{\star}_{\rm ESY}(F)$ as the expected sustainable yield at stationarity:
\begin{equation}
Y^{\star}_{\rm ESY}(F) := \E(Y^*_F) = \int_{y=0}^{\infty} yf_{Y^{\star}_F}(y|F)dy.
\end{equation}
We define the maximum expected sustainable yield $Y^{\star}_{\rm MESY}$ as
\begin{equation}
Y^{\star}_{\rm MESY} = \max_{F\geq0}Y^{\star}_{\rm ESY}(F).
\end{equation}
We denote by $F_{\rm MESY}$ the fishing mortality that attains the maximum expected sustainable yield $Y^{\star}_{\rm MESY}$, that is, the value $F$ for which $Y^{\star}_{\rm MESY}=Y^{\star}_{\rm ESY}(F)$.
\end{definition}

\begin{definition}
Define $Y^{\star}_{\rm ELSY}=Y^{\star}_{\rm ELSY}(F)$ as the expected log-sustainable yield at stationarity:
\begin{equation}
Y^{\star}_{\rm ELSY}(F) := \exp\left\{\E(\log(Y^*_F))\right\} = \exp\left\{\int_{y=0}^{\infty} \log(y)f_{Y^{\star}_F}(y|F)dy\right\}
\end{equation}
We define the maximum expected log-sustainable yield $Y^{\star}_{\rm MELSY}$ as
\begin{equation}
Y^{\star}_{\rm MELSY} = \max_{F\geq0}Y^{\star}_{\rm ELSY}(F).
\end{equation}
We denote by $F_{\rm MELSY}$ the fishing mortality that attains the maximum expected log-sustainable yield $Y^{\star}_{\rm MELSY}$, that is, the value $F$ for which $Y^{\star}_{\rm MELSY}=Y^{\star}_{\rm ELSY}(F)$.
\end{definition}

\begin{definition}
Define $Y^{\star}_{\rm EHSY}=Y^{\star}_{\rm EHSY}(F)$ as the expected harmonic sustainable yield at stationarity:
\begin{equation}
Y^{\star}_{\rm EHSY}(F) := \E((Y^*_F)^{-1})^{-1} = \left(\int_{y=0}^{\infty} y^{-1}f_{Y^{\star}_F}(y|F)dy\right)^{-1}.
\end{equation}
We define the maximum expected harmonic sustainable yield $Y^{\star}_{\rm MEHSY}$ as
\begin{equation}
Y^{\star}_{\rm MEHSY} = \max_{F\geq0}Y^{\star}_{\rm EHSY}(F).
\end{equation}
We denote by $F_{\rm MEHSY}$ the fishing mortality that attains the maximum expected harmonic sustainable yield $Y^{\star}_{\rm MEHSY}$, that is, the value $F$ for which $Y^{\star}_{\rm MEHSY}=Y^{\star}_{\rm EHSY}(F)$.
\end{definition}
}\fi

\begin{definition}
Define $Y^{\star}_{\rm ESY}(F)$ (resp. $Y^{\star}_{\rm ELSY}(F)$, $Y^{\star}_{\rm EHSY}(F)$) as the expected sustainable (resp. log-sustainable, harmonic sustainable) yield at stationarity:
\begin{equation}
\begin{split}
Y^{\star}_{\rm ESY}(F) &:= \E(Y^*_F) = \int_{y=0}^{\infty} yf_{Y^{\star}_F}(y|F)dy,\\
Y^{\star}_{\rm ELSY}(F) &:= \exp\left\{\E(\log(Y^*_F))\right\} = \exp\left\{\int_{y=0}^{\infty} \log(y)f_{Y^{\star}_F}(y|F)dy\right\},\\
Y^{\star}_{\rm EHSY}(F) &:= \E((Y^*_F)^{-1})^{-1} = \left(\int_{y=0}^{\infty} y^{-1}f_{Y^{\star}_F}(y|F)dy\right)^{-1}.
\end{split}
\end{equation}

We define the maximum expected sustainable (resp. log-sustainable, harmonic sustainable) yield $Y^{\star}_{\rm MESY}$ (resp. $Y^{\star}_{\rm MELSY}$, $Y^{\star}_{\rm MEHSY}$) as
\begin{equation}
\begin{split}
Y^{\star}_{\rm MESY} := \max_{F\geq0}\,Y^{\star}_{\rm ESY}(F),\quad Y^{\star}_{\rm MELSY} := \max_{F\geq0}\,Y^{\star}_{\rm ELSY}(F),\quad Y^{\star}_{\rm MEHSY} := \max_{F\geq0}\,Y^{\star}_{\rm EHSY}(F).
\end{split}
\end{equation}

We denote by $F_{\rm MESY}$ (resp. $F_{\rm MELSY}$, $F_{\rm MEHSY}$) the fishing mortality that attains the maximum expected sustainable (resp. log-sustainable, harmonic sustainable) yield $Y^{\star}_{\rm MESY}$ (resp. $Y^{\star}_{\rm MELSY}$, $Y^{\star}_{\rm MEHSY}$).
\end{definition}

\begin{remark}
For any fixed $F\geq0$, we have the inequality $Y^{\star}_{\rm EHSY}(F)\leq Y^{\star}_{\rm ELSY}(F)\leq Y^{\star}_{\rm ESY}(F)$. Indeed, both inequalities are a consequence of Jensen's inequality \cite[Lemma~2.5]{kallenberg}. This shows that the sustainable yield measures previously defined express a different degree of neutrality or aversion to risk. Consequently, $Y^{\star}_{\rm MEHSY}\leq Y^{\star}_{\rm MELSY}\leq Y^{\star}_{\rm MESY}$.
\end{remark}

{\color{blue}{



\if{

\section{Simplified model and effect of stochasticity}

Consider a simplified version of the model consisting of only one age category and that fish spawn during their first year of life, and then they leave the population. This situation is modeled by the equation
\begin{equation}
N_{t+1} = \varphi(e^{-Z}N_t) U_t,
\end{equation}
where $Z=M+F$ is the mortality rate, $\varphi(\cdot)$ is the Beverton-Holt spawner-recruitment funcion $\varphi(s)=\alpha s/(\beta + s)$, and $U_t=\exp\{ \sigma W_t -\sigma^2/2 \}$, where $(W_t)_{t\in\N}$ is a sequence of independent standard normal variables and $\sigma>0$ represents the variability of the recruitment. We propose the following proposition:
\begin{proposition}
For any initial probability distribution of $N_0$ and $t\geq1$, the function $\sigma\mapsto \E_{N_0}\left[ N_t \right]$ is decreasing.
\end{proposition}

\begin{proof}
We start by providing a closed formula for the recurrence:
\begin{equation}
N_t = \frac{\alpha^t\prod_{j=0}^{t-1}U_j}{(\beta e^Z)^t+N_0\sum_{j=0}^{t-1}(\beta e^Z)^{t-1-j}\alpha^{j}\prod_{k=0}^{j-1}U_k }N_0,
\end{equation}
(it is possible to prove it by induction). The expectation for $t=1$ is 
\begin{equation*}
\begin{split}
\E\left[ N_1 \right] &= \E\left[\frac{\alpha N_0 U_0}{\beta e^Z+N_0}  \right]\\
&= \E\left[\E\left[\frac{\alpha N_0 }{\beta e^Z+N_0} U_0 \big| N_0 \right]\right]\\
&= \E\left[ \frac{\alpha N_0 }{\beta e^Z+N_0} \E\left[ U_0 \right]\right]\\
&= \E\left[ \frac{\alpha N_0 }{\beta e^Z+N_0} \right]\\
\end{split}
\end{equation*}
which is constant for $\sigma$. For $t=2$ we have
\begin{equation*}
\begin{split}
\E\left[ N_2 \right] &= \E\left[ \frac{\alpha^2 N_0U_0U_1 }{(\beta e^Z)^2 + (\beta e^Z) + \alpha N_0U_0}  \right]\\
&= \E\left[\E\left[  \frac{\alpha^2 N_0U_0U_1 }{(\beta e^Z)^2 + (\beta e^Z)N_0 + \alpha N_0U_0}  \big| N_0,\,U_0 \right]\right]\\
&= \E\left[ \frac{\alpha^2 N_0U_0 }{(\beta e^Z)^2 + (\beta e^Z)N_0 + \alpha N_0U_0}  \E\left[ U_1 \right] \right]\\
&= \E\left[ \frac{\alpha^2 N_0U_0 }{(\beta e^Z)^2 + (\beta e^Z)N_0 + \alpha N_0U_0}  \right]\\
&= \E\left[ \frac{\alpha U_0 }{(\beta e^Z)^2/(\alpha N_0) + (\beta e^Z)/\alpha + U_0}  \right]\\
&= \E\left[ \E\left[ \frac{\alpha U_0 }{(\beta e^Z)^2/(\alpha N_0) + (\beta e^Z)/\alpha + U_0} \big| N_0\right] \right]\\
&= \E\left[ \E\left[ \frac{\alpha U_0 }{\tilde \beta(N_0) + U_0} \big| N_0\right] \right]\\
\end{split}
\end{equation*}

We now study, for fixed $\alpha,\beta>0$, the expectation of the random variable $T:=\alpha U/(\beta+U)$ depending of $\sigma$, where $U=\exp\{\sigma W-\sigma^2/2\}$ and $W\sim N(0,1)$. More precisely, in what follows, we prove that the function $g(\sigma):=\E^{\sigma}\left[T\right]$ is decreasing. If we define $\phi(\omega):=\frac{1}{\sqrt{2\pi}}\exp\{-\omega^2/2\}$ as the standard Normal density function, we have
\begin{equation}
\begin{split}
g(\sigma) \,=\, \E^{\sigma}\left[T\right] &= \E^{\sigma}\left[\frac{\alpha U}{\beta+U}\right] \\
&= \E^{\sigma}\left[\frac{\alpha e^{\sigma W-\sigma^2/2}}{\beta+e^{\sigma W-\sigma^2/2}}\right] \\
&= \E^{\sigma}\left[\frac{\alpha e^{\sigma W}}{\beta e^{\sigma^2/2}+e^{\sigma W}}\right] \\
&= \int_{\R}\frac{\alpha e^{\sigma\omega}}{\beta e^{\sigma^2/2}+e^{\sigma\omega}}\phi(\omega)d\omega \\
\end{split}
\end{equation}

Consider the integrand function 
\begin{equation*}
h(\omega,\sigma) := \frac{\alpha e^{\sigma\omega}}{\beta e^{\sigma^2/2}+e^{\sigma\omega}}\phi(\omega),
\end{equation*}
which is continuous and differentiable, and its partial derivative with respect to $\sigma$ is
\begin{equation*}
\frac{\partial h}{\partial \sigma}(\omega,\sigma) = \frac{\alpha e^{\sigma\omega}\beta e^{\sigma^2/2}}{(\beta e^{\sigma^2/2}+e^{\sigma\omega})^2}(\omega-\sigma)\phi(\omega), 
\end{equation*}
which is a continuous function, and for $\sigma_{\rm max}>0$ large enough, it is uniformly bounded by an integrable function. Indeed,
\begin{equation*}
\left|\frac{\partial h}{\partial\sigma}(\omega,\sigma)\right| \leq \alpha(|\omega|+\sigma_{\rm max})\phi(\omega),
\end{equation*}
Then, for $\sigma\in[0,\sigma_{\rm max}]$, we have
\begin{equation*}
\begin{split}
\frac{\partial}{\partial\sigma}\E^{\sigma}\left[T\right] &= \int_{\R} \frac{\partial h}{\partial \sigma}(\omega,\sigma) d\omega \,=\, \int_{\R}\frac{\alpha e^{\sigma\omega}\beta e^{\sigma^2/2}}{(\beta e^{\sigma^2/2}+e^{\sigma\omega})^2}(\omega-\sigma)\phi(\omega) d\omega, 
\end{split}
\end{equation*}

Now we consider, for fixed $\sigma\geq0$, the function $\rho_\sigma(\omega)=\frac{\partial h}{\partial\sigma}(\omega,\sigma)$. The unique zero of this function is $\omega^{\star}:=\sigma$; for $\omega<\sigma$ we have $\rho_{\sigma}(\omega)<0$ and for $\omega>\sigma$ we have $\rho_{\sigma}(\omega)>0$. Now, 
\begin{equation*}
\begin{split}
\frac{\partial}{\partial\sigma}\E^{\sigma}\left[T\right] &= \int_{-\infty}^{\sigma} \rho_{\sigma}(\omega)d\omega + \int_{\sigma}^{\infty} \rho_{\sigma}(\omega)d\omega\\
&= \int_{-\infty}^{0} \rho_{\sigma}(\omega+\sigma)d\omega + \int_{0}^{\infty} \rho_{\sigma}(\omega+\sigma)d\omega\\
&= \int_{-\infty}^{0} \rho_{\sigma}(\sigma+\omega)d\omega + \int_{-\infty}^{0} \rho_{\sigma}(\sigma-\omega)d\omega\\
&= \int_{-\infty}^{0} \left[ \rho_{\sigma}(\sigma+\omega) + \rho_{\sigma}(\sigma-\omega) \right]d\omega\\
\end{split}
\end{equation*}
where, for $\sigma<0$, 
\begin{equation*}
\displaystyle\begin{split}
\rho_{\sigma}(\sigma+\omega) + \rho_{\sigma}(\sigma-\omega) &= \alpha\beta e^{\sigma^2/2} \left[\frac{e^{\sigma(\sigma+\omega)}((\sigma+\omega)-\sigma)}{(\beta e^{\sigma^2/2}+e^{\sigma(\sigma+\omega)})^2}\phi(\sigma+\omega)
 + \frac{ e^{\sigma(\sigma-\omega)}((\sigma-\omega)-\sigma)}{(\beta e^{\sigma^2/2}+e^{\sigma(\sigma-\omega)})^2}\phi(\sigma-\omega)\right]\\
&= \frac{\alpha\beta e^{\sigma^2/2}e^{\sigma^2}}{\sqrt{2\pi}}\omega \left[\frac{e^{\sigma\omega}e^{-\frac{1}{2}(\sigma+\omega)^2}}{(\beta e^{\sigma^2/2}+e^{\sigma(\sigma+\omega)})^2}
 - \frac{ e^{-\sigma\omega}e^{-\frac{1}{2}(\sigma-\omega)^2}}{(\beta e^{\sigma^2/2}+e^{\sigma(\sigma-\omega)})^2}\right]\\
&= \frac{\alpha\beta e^{\sigma^2/2}e^{\sigma^2}}{\sqrt{2\pi}}\omega \left[\frac{e^{\sigma\omega}e^{-\frac{\sigma^2}{2}-\frac{\omega^2}{2}-\sigma\omega}}{(\beta e^{\sigma^2/2}+e^{\sigma(\sigma+\omega)})^2}
 - \frac{ e^{-\sigma\omega}e^{-\frac{\sigma^2}{2}-\frac{\omega^2}{2}+\sigma\omega}}{(\beta e^{\sigma^2/2}+e^{\sigma(\sigma-\omega)})^2}\right]\\
&= \frac{\alpha\beta e^{\sigma^2/2}e^{\sigma^2}}{\sqrt{2\pi}}\omega e^{-\frac{\sigma^2}{2}-\frac{\omega^2}{2}} \left[\frac{1}{(\beta e^{\sigma^2/2}+e^{\sigma(\sigma+\omega)})^2}
 - \frac{1}{(\beta e^{\sigma^2/2}+e^{\sigma(\sigma-\omega)})^2}\right]\\
&= \frac{\alpha\beta e^{\sigma^2}}{\sqrt{2\pi}}\omega e^{-\frac{\omega^2}{2}} \frac{(\beta e^{\sigma^2/2}+e^{\sigma(\sigma-\omega)})^2 - (\beta e^{\sigma^2/2}+e^{\sigma(\sigma+\omega)})^2}{(\beta e^{\sigma^2/2}+e^{\sigma(\sigma+\omega)})^2(\beta e^{\sigma^2/2}+e^{\sigma(\sigma-\omega)})^2} \\
&= \frac{\alpha\beta e^{\sigma^2}}{\sqrt{2\pi}}\omega e^{-\frac{\omega^2}{2}} \frac{(e^{\sigma(\sigma-\omega)} - e^{\sigma(\sigma+\omega)})( 2\beta e^{\sigma^2/2} + e^{\sigma(\sigma-\omega)}+e^{\sigma(\sigma+\omega)})}{(\beta e^{\sigma^2/2}+e^{\sigma(\sigma+\omega)})^2(\beta e^{\sigma^2/2}+e^{\sigma(\sigma-\omega)})^2} \\
&= \frac{\alpha\beta e^{\sigma^2}}{\sqrt{2\pi}}\omega e^{-\frac{\omega^2}{2}}e^{\sigma^2} \frac{(e^{-\sigma\omega} - e^{\sigma\omega})( 2\beta e^{\sigma^2/2} + e^{\sigma(\sigma-\omega)}+e^{\sigma(\sigma+\omega)})}{(\beta e^{\sigma^2/2}+e^{\sigma(\sigma+\omega)})^2(\beta e^{\sigma^2/2}+e^{\sigma(\sigma-\omega)})^2} \\
&= \frac{\alpha\beta e^{\sigma^2}}{\sqrt{2\pi}} e^{-\frac{\omega^2}{2}}e^{\sigma^2} \frac{( 2\beta e^{\sigma^2/2} + e^{\sigma(\sigma-\omega)}+e^{\sigma(\sigma+\omega)})}{(\beta e^{\sigma^2/2}+e^{\sigma(\sigma+\omega)})^2(\beta e^{\sigma^2/2}+e^{\sigma(\sigma-\omega)})^2} (e^{-\sigma\omega} - e^{\sigma\omega})\omega.
\end{split}
\end{equation*}
In this last expression, all the terms are positive except for $(e^{-\sigma\omega} - e^{\sigma\omega})\omega$. Since in the integral we consider $\omega\leq0$, this last term is negative, obtaining that $\frac{\partial}{\partial\sigma}\E^{\sigma}T<0$ and concluding that $g(\cdot)$ is decreasing.

\end{proof}

\medskip

}\fi

}}

\section{Illustration: Patagonian toothfish}
\label{sec:numSim}

Patagonian toothfish (\emph{dissostichus eleginoides}) is a fish species that lives in the southern Pacific and Atlantic oceans (Figure \ref{fig:patagonia}). It was first researched as a potential fishery resource in Chile in the 1950s by exploratory fishings at low depths \cite{guerrero,moreno}. Thirty years later, this resource began to be caught as bycatch in trawl fisheries developed around Kerguelen island, Patagonic platform and the South Georgia Islands. In the mid 1980s, development of longlines that could be operated at low depths led to fishing in Chilean waters directed towards adult individuals. Fishing activity quickly expanded to the Patagonic platform, South Georgia, and Kerguelen. The high price of this product in the international market led to large increases in catch and the exploration of new fishing grounds. According to FAO, and including only legal catches, landings in CCAMLR and territorial waters, catches increased from less than 5.000 tons in 1983 to over 40.000 tons en 1992 \cite{bacalao}.

In Argentina, the patagonian toothfish fishery followed a similar development, starting as bycatch in trawl fisheries, and subsequently developing longline fisheries in Argentina as well as in Falklands. The Argentinean fishery started in the 1990s and reached its peak in 1995 with 19.225 tons; since then, catches have decreased. Longline fishing started in the Falklands as an experimental fishery in 1992 and became established in 1994 \cite{laptikhovsky}. The catch reached a maximum of 2.733 tons in 1994 and then it become stable in the range of 1.200 - 1.800 tons.

\begin{center}
\includegraphics[scale=0.35]{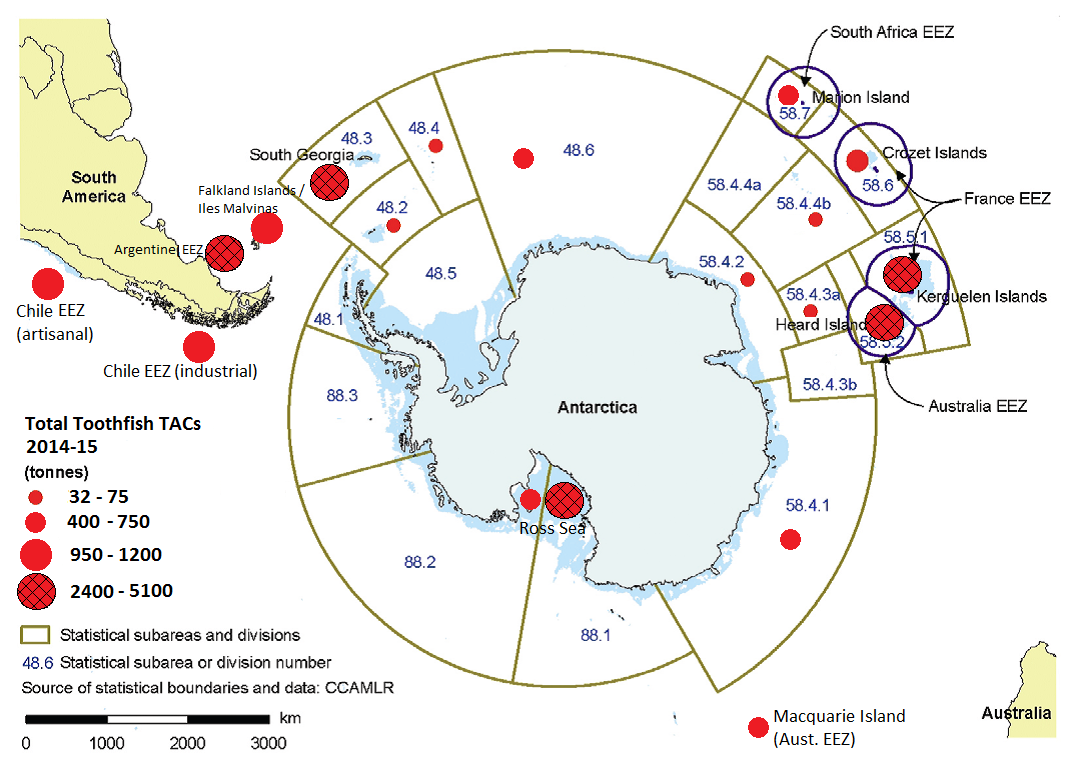}
\captionof{figure}{Distribution of patagonian toothfish and Antarctic toothfish in the southern Pacific and Atlantic Oceans. Source: ccamlr.org}\label{fig:patagonia}
\end{center}

On the basis in recomendations of Instituto Nacional de Investigaci\'on y Desarrollo Pesquero de Argentina (INIDEP), since the year 2000 the size of hooks is regulated, catches are documented, a minimum size limit is in place, and there are minimum depths of operation and a protection area for young individuals.

The Chilean patagonian toothfish fishery is divided mainly in two zones: the north zone, between the northern limit of the country (18\textdegree 21\prim) and parallel 47\textdegree, is reserved exclusively for the artisanal fleet. In the south zone (47\textdegree S - 57\textdegree S), the industrial fleet mainly operates.

The Argentinean patagonian toothfish fishery is comprised of two fleets distinguished by their rigging/fishing gear used: the longline fleet started operating in 1990 and since its inception, has been a directed fishery with an area of operation involving almost the entire ranges of the resource in the Argentina platform. The longline fishery is responsible for the largest historic landing registered in 1995, from which catches were significantly reduced by this fleet. The number of ships that make up the longline fleet has been on a gradual decline from a peak of 25 in 1996 to 4 in 2013 \cite{wohler}. 
The fleet operating with bottom trawling began in the late 1980s. Because of the differential size distribution with depth exhibited by toothfish, and because most of the trawl sets are made between 400 and 500 [m] deep, the trawl fleet mainly impacts the juvenile fraction of the population. The catches of trawlers showed an increasing trend from 1999, which is related to the exploration of new fishing areas and not to an increase in the resource abundance. Currently this fleet is composed of 5 trawlers \cite{wohler}.\smallskip

In Chile, in the year 2013, the new General Law of Fisheries and Aquaculture (Ley General de Pesca y Acuicultura, LGPA) took effect. Regarding fisheries, the modifications to the law covered five fundamental aspects: sustainability, industrial and artisanal fishery regulations, research, and audit. One of the main aspects will be to keep or to rebuild the fishery to the maximum sustainable yield (MSY), considering the biological characteristics of the exploited resources. It is important to have good estimate of MSY, which can be used as a target or limit for the harvest control policy.\smallskip

We model the Patagonian toothfish population as an age-structured population with a plus group at age $A=30$ years, and the age of recruitment $r=3$ years. Spawning occurs at $\tau=7/12$ of the way through the year. There are four fleets: Chilean Industrial fleet, Chilean Artisanal fleet, Argentinean longline fleet, and Argentinean Artisanal fleet. First, we investigate the estimation of model parameters and outputs, including the MSY for a deterministic system, that is, considering no effect of stochasticity, and then we investigate how the introduction of stochasticity on recruitment affects the behavior of the system by estimating the values of MESY, MELSY, and MEHSY.
\smallskip

For the numerical simulations, we consider data of the landings from 1978 to 2014 and the parameters estimated from stock assesment obtained from IFOP webpage (more details in \cite{bacalao}). The parameter values, obtained from \cite{bacalao}, are $M=0.15$, $h = 0.6$, $R_0 = 5309$, $B_0 = 214009$, which give as a result $\alpha = 6370.8$, $\beta = 42801.8$. Also, the proportions of fishing mortality of the different fleets are $P = (50.23\%,\, 23.57\%,\, 16.44\%,\, 9.76\%)$, so that the Chilean industrial fleet currently bytakes a majority of the fish, followed by Chilean artisanal fleet, Argentinean longline fleet, and Argentinean artisanal fleet. Selectivities are shown in Table \ref{table:app1} and initial conditions are shown in Table \ref{table:app2} in the Appendix \ref{app:2}.\smallskip

We consider a time horizon $T_{\rm end}=200$ years (time in which we can observe a steady behavior for the deterministic trajectories). As a first step, we run simulations of the deterministic process \eqref{eq:modelBacalao}. We define a meshgrid in the interval $[0,1]$ with a step $h_d=10^{-5}$, and for each value $F_j$ in this meshgrid we compute the theoretical equilibrium values $N^{\star}(F_j)$, ${\rm SSB}^{\star}(F_j)$ and $Y^{\star}(F)$ as given in equations \eqref{eq:equilib_N}, \eqref{eq:equilib_N1}, and \eqref{eq:equilib_Y}. The obtained value $F=F_{\rm MSY}$ that maximizes the sustainable yield is $F_{\rm MSY} = 0.13$ and the corresponding maximum yield is $Y^{\star}_{\rm MSY}=7324 [{\rm tons}]$ (see Figure \ref{fig:01}). The simulated values match the theoretical values.

\begin{center}
\begin{minipage}{8cm}
\begin{center}
\vspace{4mm}
\includegraphics[scale=0.23]{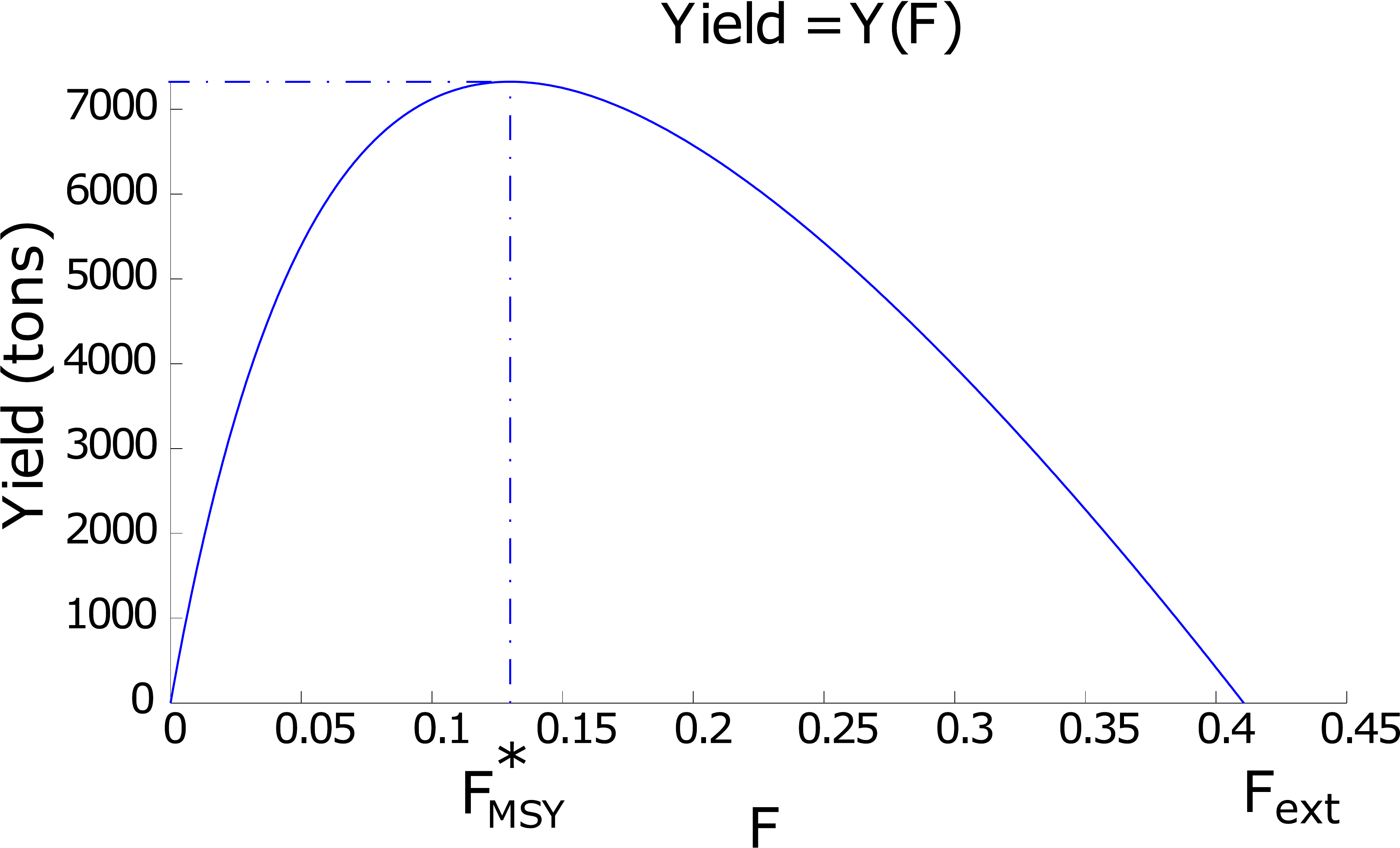}
\captionof{figure}{Deterministic yield as function of fishing mortality.}\label{fig:01}
\end{center}
\end{minipage}
\ \
\begin{minipage}{8cm}
\begin{center}
\includegraphics[scale=0.23]{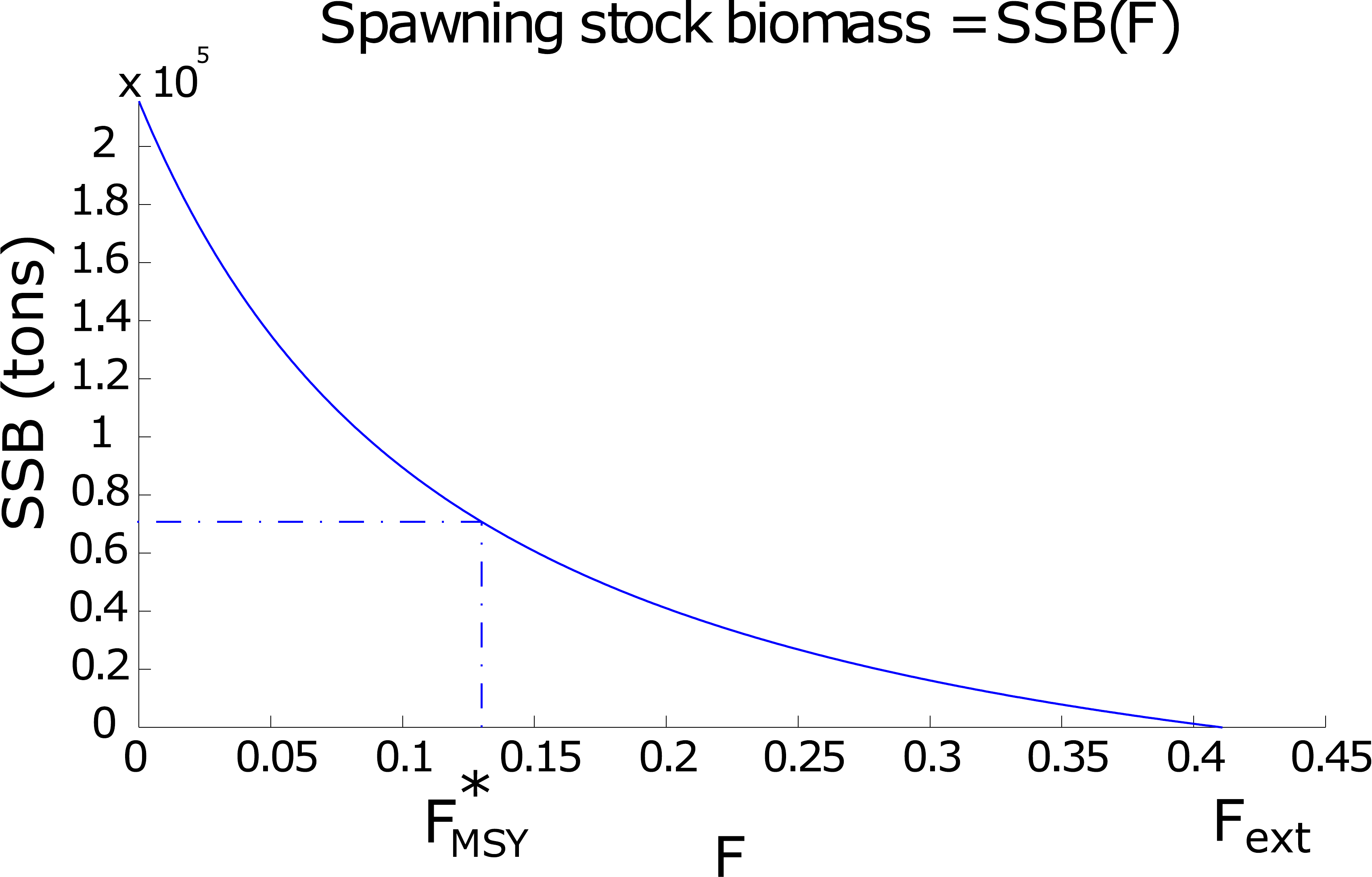}
\captionof{figure}{Deterministic SSB as function of fishing mortality.}\label{fig:02}
\end{center}
\end{minipage}
\end{center}

We can see that the yield reaches 0 at about $F_{\rm ext}=0.41$. This coincides with the fact that for $F>0.41$ the relation ${\rm SPR}(F)>\frac{\beta}{\alpha}$ is no longer satisfied, and then the corresponding theoretical equilibrium points become negative (as well as the yield and SSB). This relations are shown in Figures \ref{fig:01}-\ref{fig:03}.
\begin{center}
\includegraphics[scale=0.30]{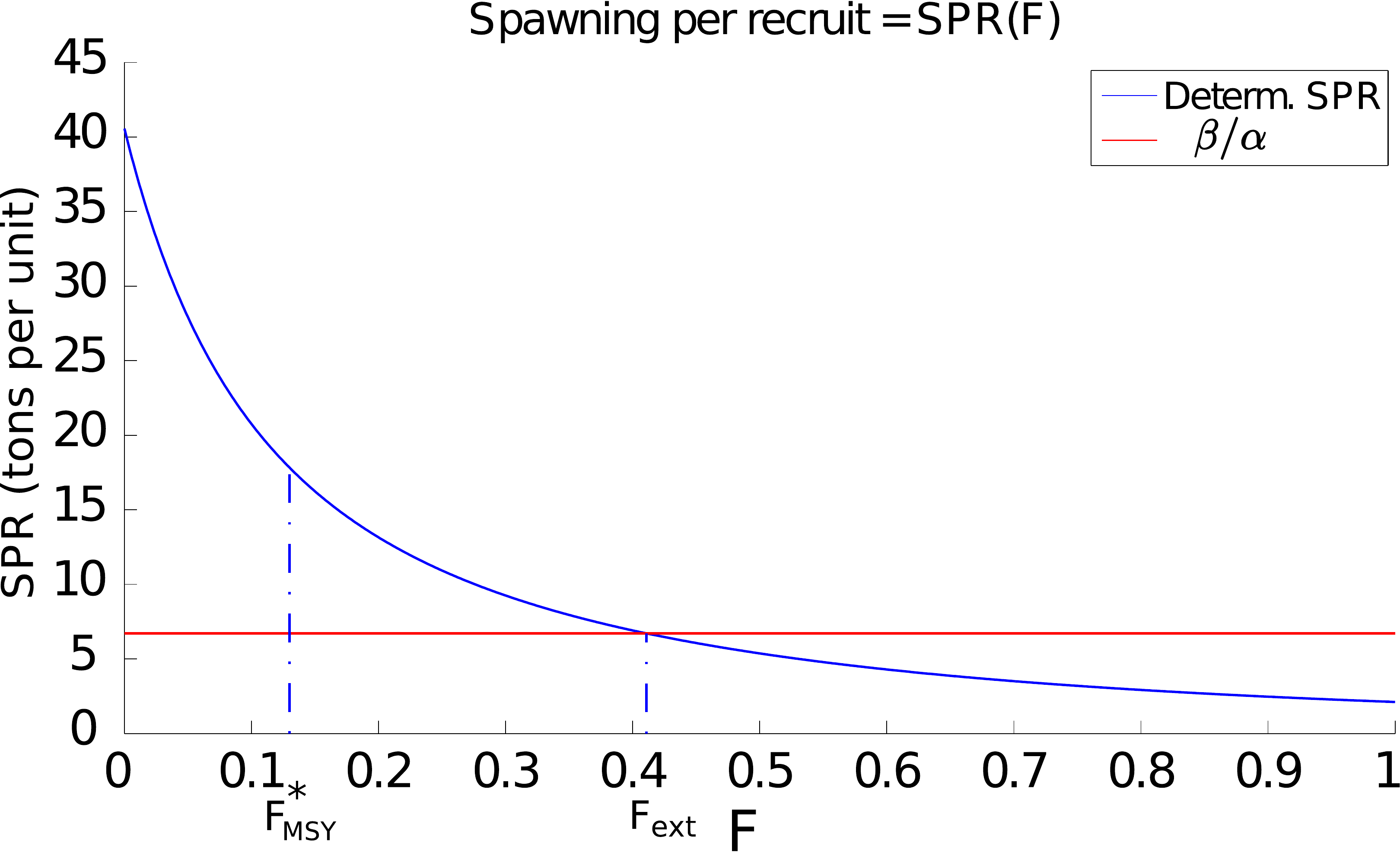}
\captionof{figure}{Deterministic SPR as function of fishing mortality.}\label{fig:03}
\end{center}
\medskip

The expected values for the process with stochasticity in recruitment must be determined by numerical simulations. In order to find $F_{\rm MESY}$, $F_{\rm MELSY}$, and $F_{\rm MEHSY}$, 
we search the optimal values for $F$ in the interval $[0,F_{\rm ext}]$ by defining a meshgrid, and for each value $F_j$ in the meshgrid we compute $\nu=500.000$ replications for different values of $CV$, namely $CV_1 = 0.25$, $CV_2=0.5$, $CV_3 = 0.75$, $CV_4 = 1$, $CV_5 = 1.5$ and $CV_6=2$. This number of replications is large enough to ensure that the confidence interval at 95\% for $Y_{\rm MESY}$ is at most of large 1\% with respect to the mean value.

The values of yield at the final time for each replication $k$ (for fixed $CV$ and $F_j$) are denoted $y_{T_{\rm end}}^{k}$. We estimate the values $Y^{\star}_{\rm ESY}(F)$, $Y^{\star}_{\rm ELSY}(F)$ and $Y^{\star}_{\rm EHSY}(F)$ by the arithmetic, geometric, and harmonic means of the final values of yield:
\begin{equation*}
\hat Y^{\star}_{\rm ESY}(F) = \frac{1}{\nu}\sum_{k=1}^{\nu} y^{k}_{T_{\rm end}},\quad
\hat Y^{\star}_{\rm ELSY}(F) = \exp\left\{\frac{1}{\nu}\sum_{k=1}^{\nu} \log\left(y^{k}_{T_{\rm end}}\right) \right\} ,\quad
\hat Y^{\star}_{\rm EHSY}(F) = \left(\frac{1}{\nu}\sum_{k=1}^{\nu} \frac{1}{y^{k}_{T_{\rm end}}} \right)^{-1}.
\end{equation*}

We search the value $\hat Y^{\star}_{\rm MESY}$ among the different values $\hat Y^{\star}_{\rm ESY}(F_j)$ (we proceed analogously for $\hat Y^{\star}_{\rm MELSY}$ and $\hat Y^{\star}_{\rm MEHSY}$). The results are shown in Table \ref{table:01}.\smallskip

\begin{center}
\captionof{table}{Values of maximum expected sustainable fishing effort, maximum expected sustainable yield, and spawning stock biomass (in tons) for different values of $CV$.}\label{table:01}
\begin{tabular}{||c||ccccccc||} \hline \hline
                  &  CV = 0   & CV = 0.25  & CV = 0.5  &  CV = 0.75 &   CV = 1   &    CV=1.5   &   CV = 2   \\ \hline \hline
$F_{\rm MESY}$                &   0.130   &  0.130  &  0.129   &   0.127  &  0.123  &   0.111  &  0.091   \\
$F_{\rm MELSY}$               &   0.130   &  0.130  &  0.128   &   0.125  &  0.120  &   0.103  &  0.078    \\
$F_{\rm MEHSY}$               &   0.130   &  0.129  &  0.127   &   0.123  &  0.117  &   0.097  &  0.068    \\ \hline
$\hat Y^{\star}_{\rm MESY}$   &    7324   & 7310    &  7264    &   7173   &   7011  &   6373   &   5141  \\ 
$\hat Y^{\star}_{\rm MELSY}$  &    7324   & 7295    &  7194    &   6997   &   6659  &   5399   &   3409  \\ 
$\hat Y^{\star}_{\rm MEHSY}$  &    7324   & 7278    &  7125    &   6829   &   6339  &   4677   &   2458  \\  \hline
$\rm SSB_{MESY}$              &   70775   & 70624   &  70717   &   70939  &  71601  &  72111   &  70944 \\ 
$\rm SSB_{MELSY}$             &   70775   & 70624   &  71271   &   72073  &  73369  &  77489   &  81772 \\ 
$\rm SSB_{MEHSY}$             &   70775   & 71173   &  71830   &   73225  &  75182  &  81787   &  91126 \\ \hline \hline
\end{tabular}
\end{center}                                                           
\medskip

The maximum expected yield (in any of its possible measures) decreases as the variability of fish recruitment increases. Also, the fishing mortality that produces the maximum expected sustainable yield decreases, which can be taken as a sign to be more cautious when there is stochasticity in recruitment present (see Table \ref{table:01}, as well as Figure \ref{fig:101}). Also, for small values of $CV$ the values $F_{\rm MESY}$, $F_{\rm MELSY}$, and $F_{\rm MEHSY}$ are close, but for large values of $CV$ they become clearly different. The most conservative measure is $F_{\rm MEHSY}$, and the least is $F_{\rm MESY}$, as expected from deterministic theory. Nevertheless, this behavior is not witnessed in the SSB at the optimal fishing mortalities, but it is possible to conclude that under a highly cautious behavior ($F_{\rm MEHSY}$), the expected SSB is by far larger than the deterministic SSB under high variability (for $CV=2$ it is a $28,5\%$ larger).
\smallskip

In Figure \ref{fig:001} we show the behavior of the estimators of the expected sustainable yield $\hat Y_{\rm ESY}(F)$ as function of $F$, for different values of $CV$. We can see that for each fixed $F$ the estimated expected sustainable yield values are decreasing with respect to $CV$, and extinction of biomass occurs for a larger range of fishing mortality than in the deterministic case. The same type of behavior can be observed for the estimator of the mean value of SSB, as Figure \ref{fig:002} shows:
\begin{figure}[h!]
\begin{center}
\begin{minipage}{8cm}
\begin{center}
\includegraphics[scale=0.23]{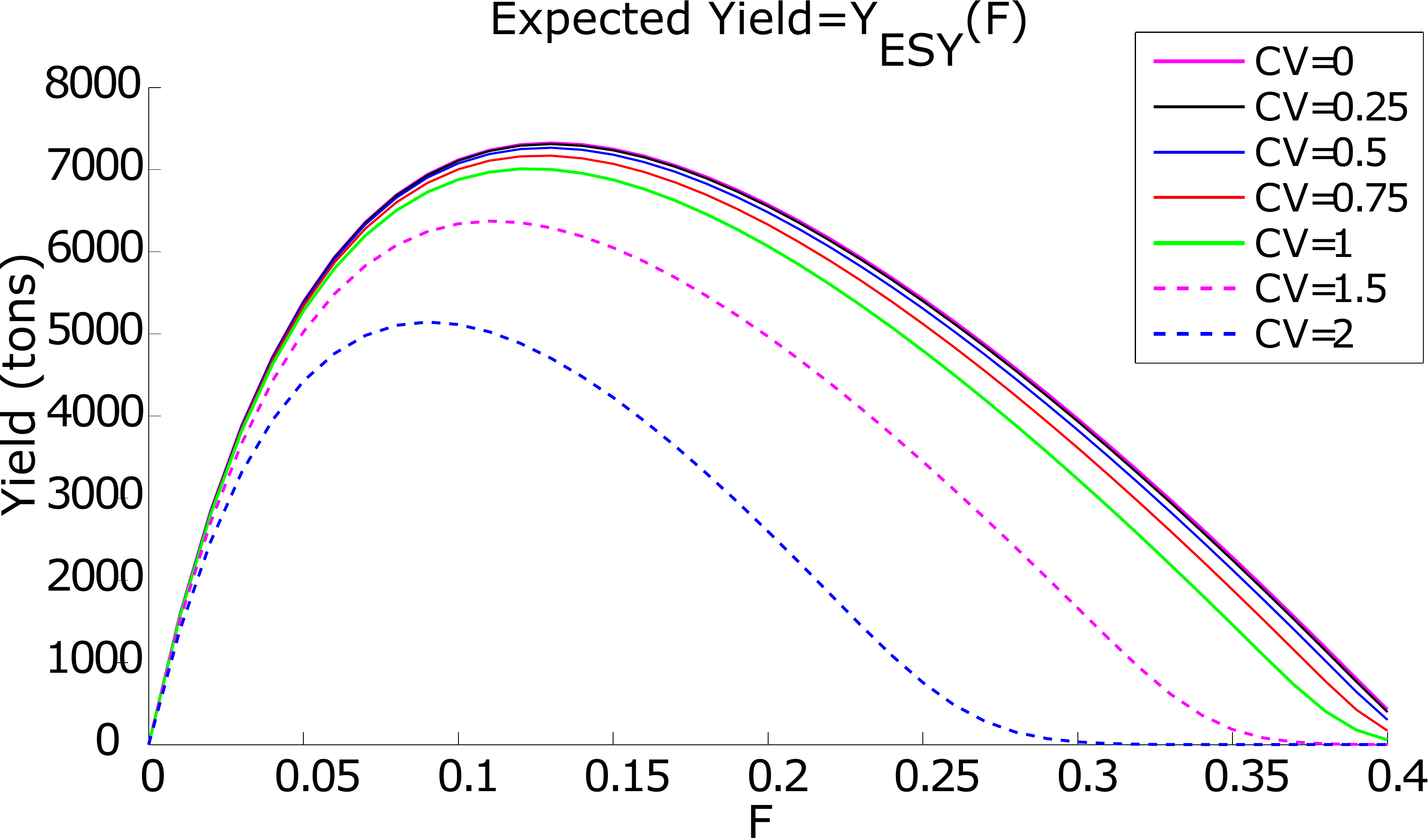}
\captionof{figure}{Yield as function of fishing mortality, for different values of $CV$.}\label{fig:001}
\end{center}
\end{minipage}
\ \
\begin{minipage}{8cm}
\begin{center}
\includegraphics[scale=0.23]{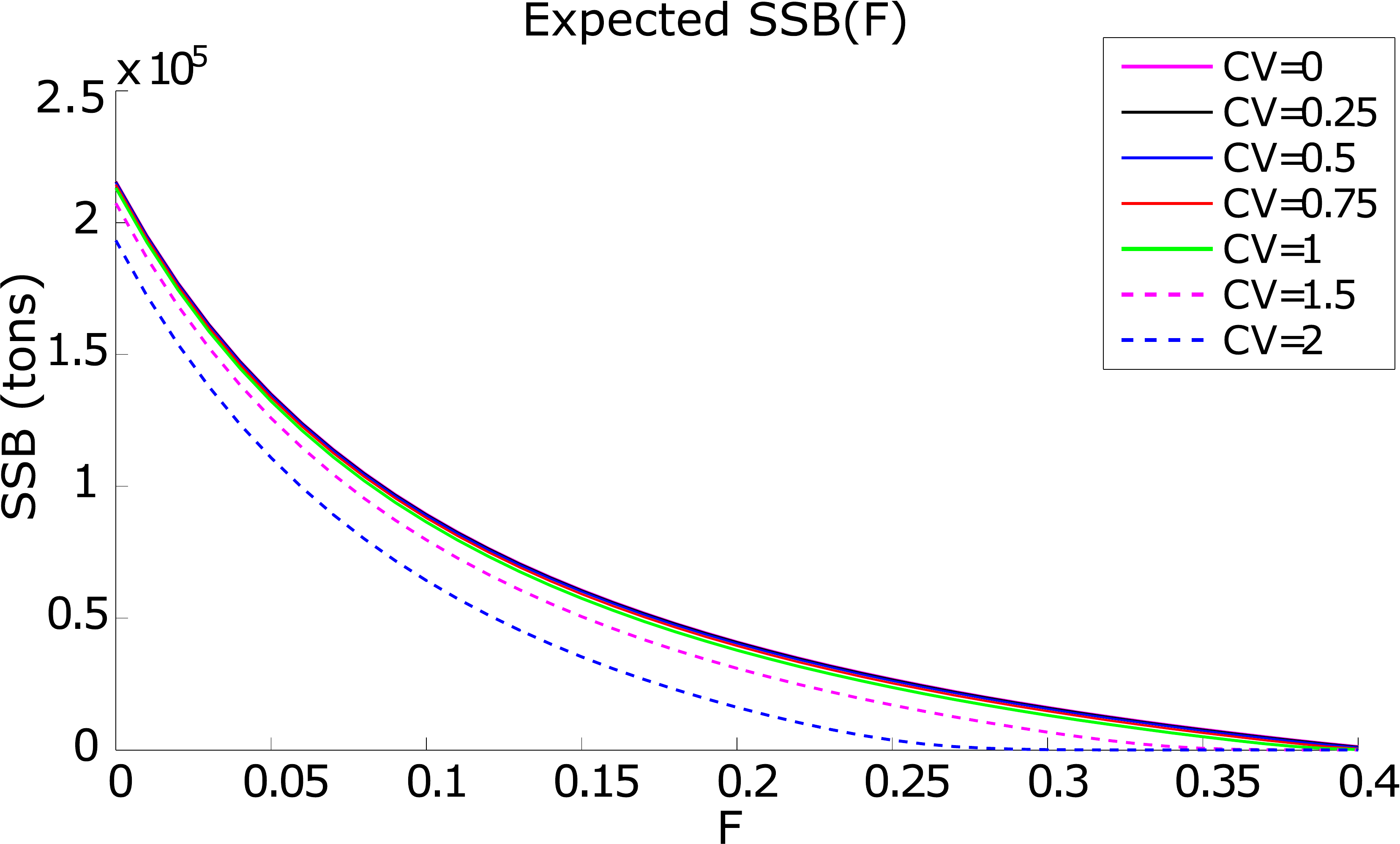}
\captionof{figure}{Spawning stock biomass as function of fishing mortality, for different values of $CV$.}\label{fig:002}
\end{center}
\end{minipage}
\end{center}
\end{figure}

In Figures \ref{fig:101} and \ref{fig:102} the behavior of the mean trajectories that lead to the maximum expected yield and the corresponding spawning stock biomass as function of time are compared to the corresponding deterministic trajectories. If $CV$ increases, the mean yield decreases from its deterministic value. This behavior is not observed for SSB. Note that each trajectory is computed using the optimal fishing mortality $F_{\rm MESY}(CV)$ corresponding to the respective value of $CV$. 
\begin{figure}[h!]
\begin{center}
\begin{minipage}{8cm}
\begin{center}
\includegraphics[scale=0.23]{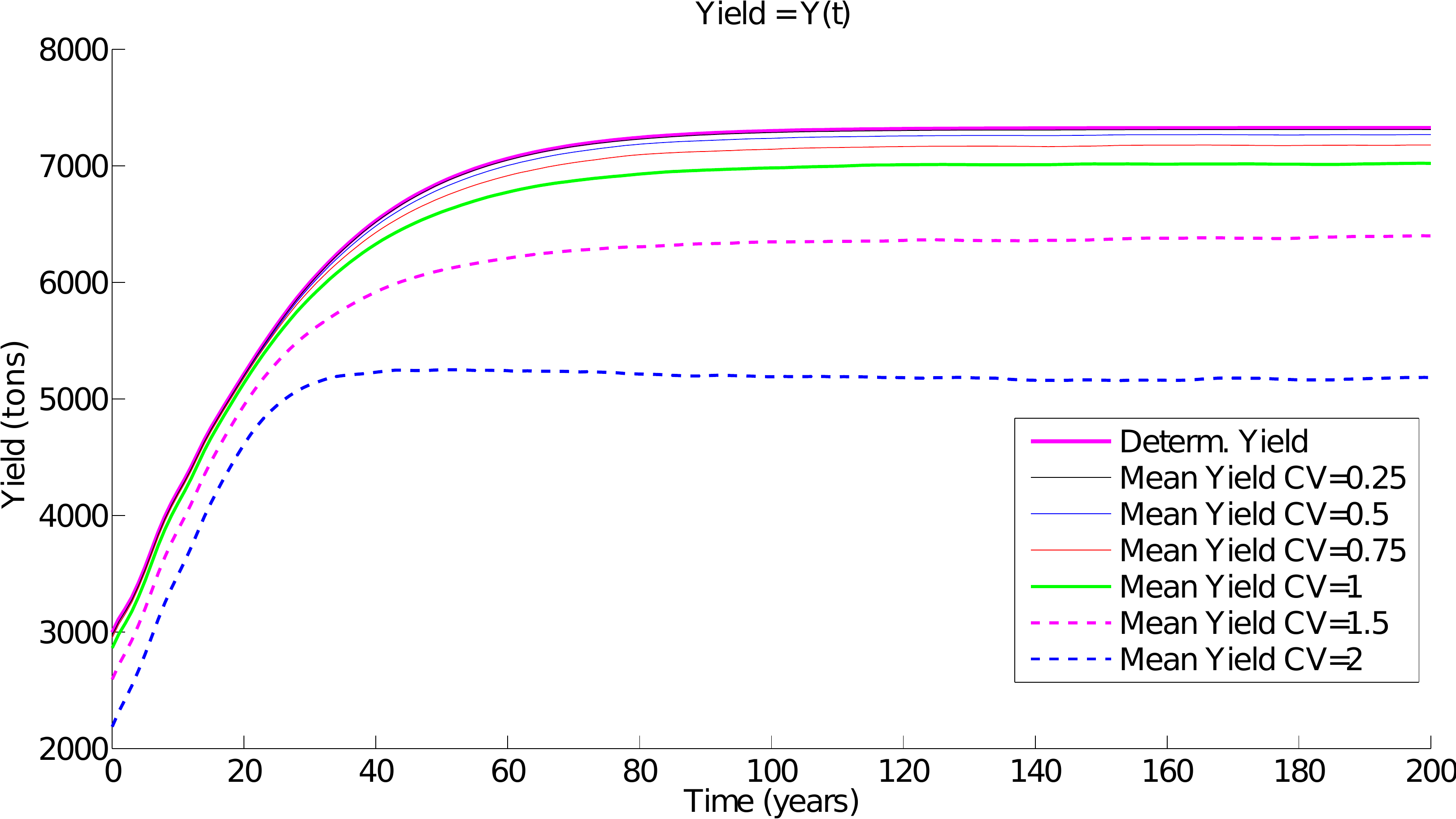}
\captionof{figure}{Deterministic and mean yield as function of time for different values of $CV$.}\label{fig:101}
\end{center}
\end{minipage}
\ \
\begin{minipage}{8cm}
\begin{center}
\vspace{4mm}
\includegraphics[scale=0.23]{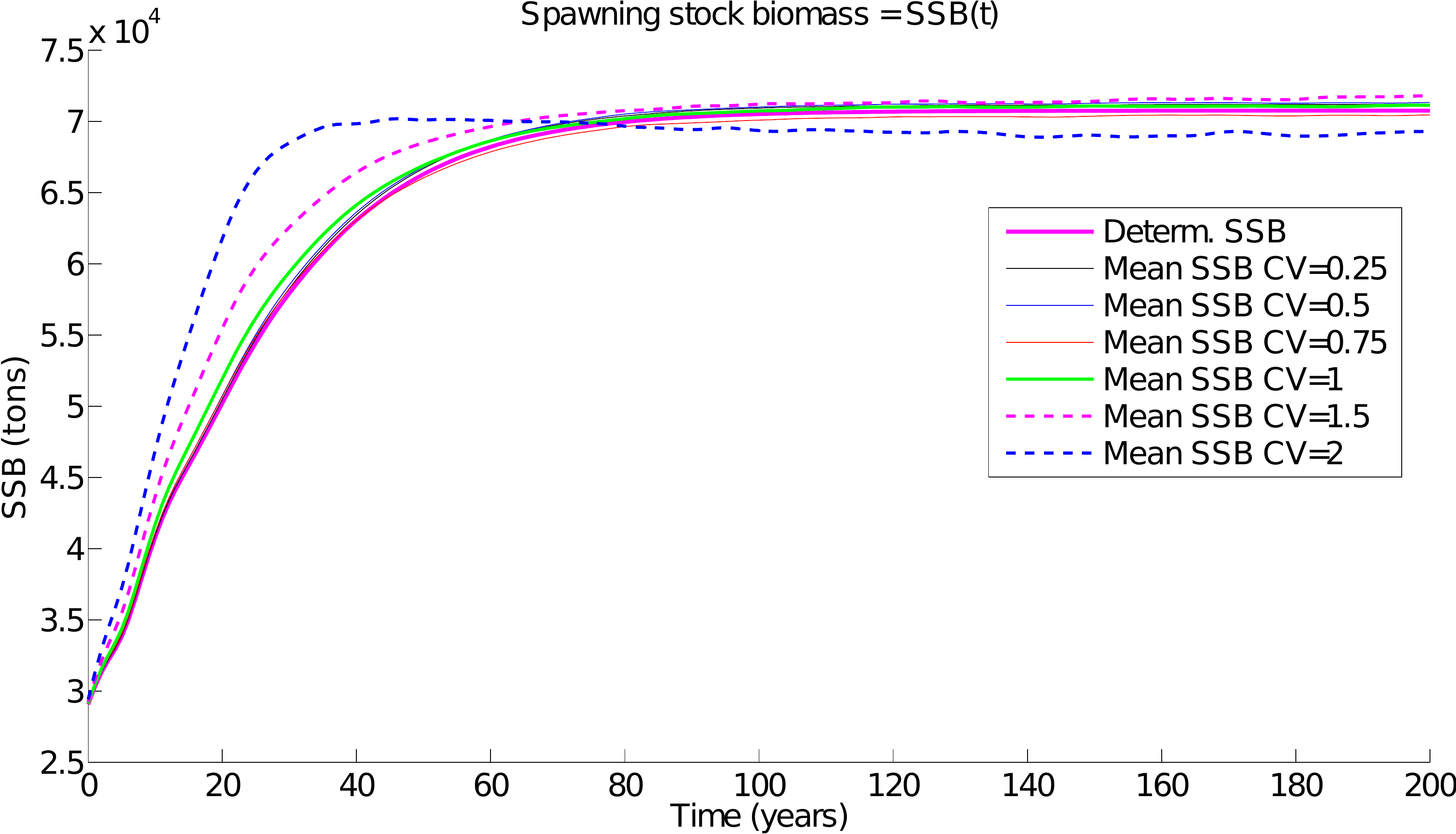}
\captionof{figure}{Deterministic and mean SSB as function of time for different values of $CV$.}\label{fig:102}
\end{center}
\end{minipage}
\end{center}
\end{figure}

Figure \ref{fig:201} shows the comparison of the probability density functions of yield (at the corresponding $F_{\rm MESY}$) for different values of $CV$. In this figure the empirical probability density function in the interval $[p_{2.5\%},p_{95\%}]$ (which is the interval in which is contained 95\% of the final values of yield, with tails of 2.5\% of the values) is shown. As the figure shows, if the coefficient of variation $CV$ increases, the yield distribution becomes more spread and moves to the left, and so do its mean values. The size of the confidence intervals also increases with $CV$, showing the necessity to consider more conservative yield measures under high volatility.
\begin{center}
\includegraphics[scale=0.35]{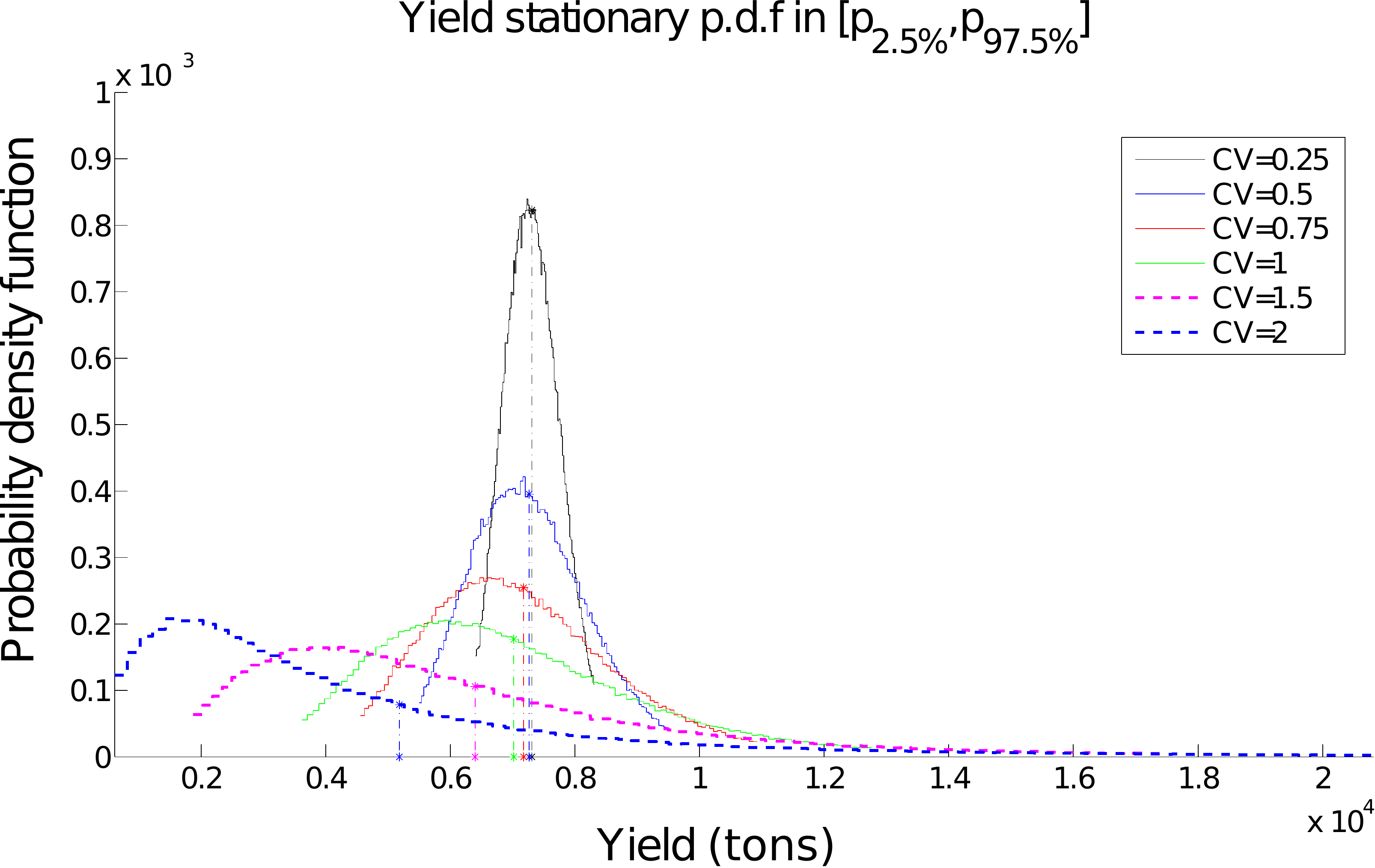}
\captionof{figure}{Probability density function of stationary yield, for different values of $CV$.}\label{fig:201}
\end{center}

In Figure \ref{fig:401} we show a comparison between the deterministic, arithmetic and geometric means of the yield as functions of the fishing mortality $F$. For small values of $CV$, the differences are not too notorious, but for large values of $CV$ they become apparent, both in the optimal mean values and in the optimal fishing mortality.
\begin{figure}[h!]
\begin{center}
\begin{minipage}{8cm}
\begin{center}
\includegraphics[scale=0.23]{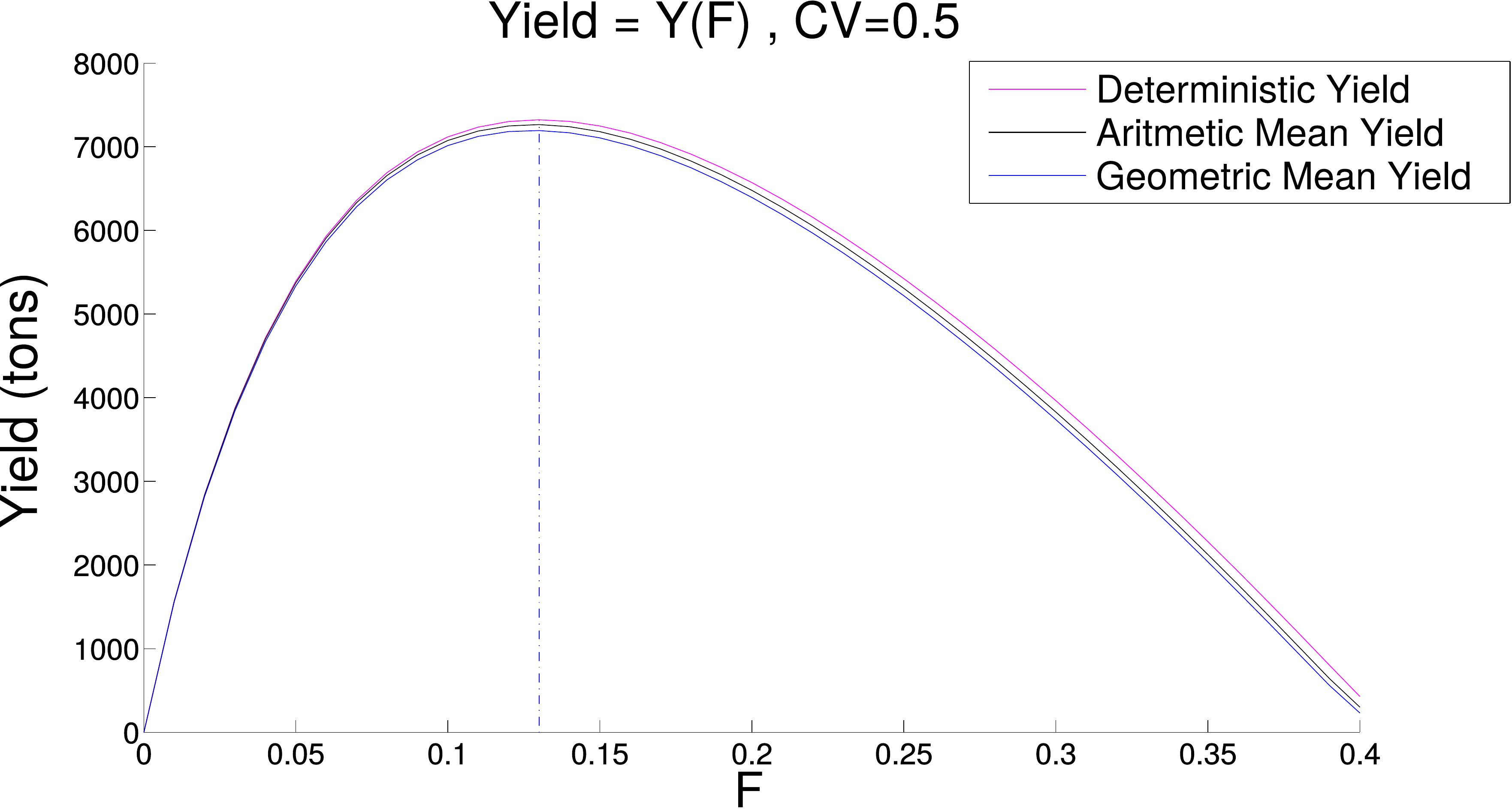}
\end{center}
\end{minipage}
\ \
\begin{minipage}{8cm}
\begin{center}
\includegraphics[scale=0.23]{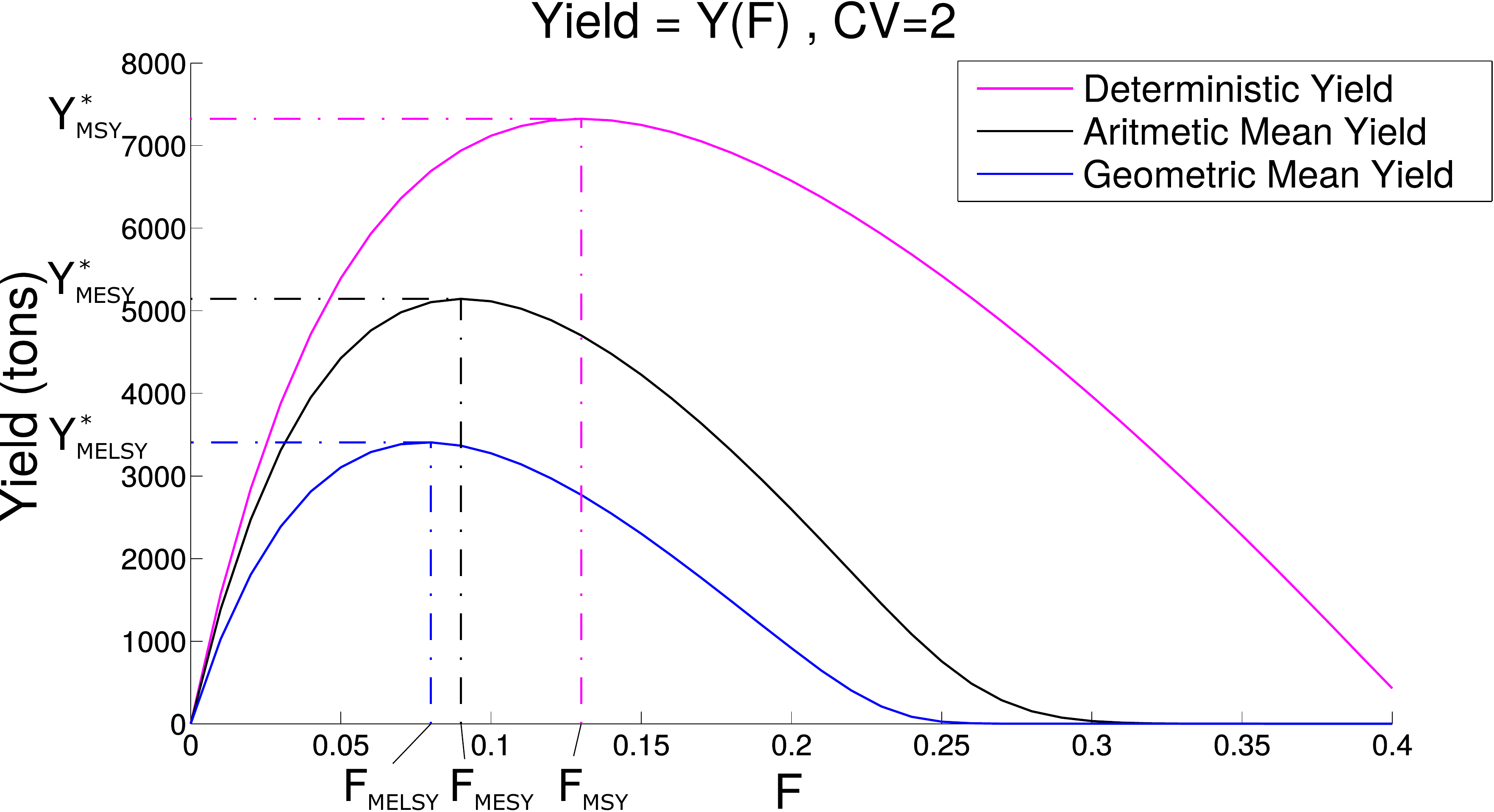}
\end{center}
\end{minipage}
\captionof{figure}{Comparison between deterministic, arithmetic and geometric yield. On the left, for $CV=0.5$, on the right, $CV=2$.}\label{fig:401}
\end{center}
\end{figure}

In Table \ref{tab:aux001} and Figure \ref{fig:aux987} the behavior of the yield and spawning stock biomass are shown for the different maximum yield measures, as well as their behavior with the current fishing mortality $F_{\rm current}=0.292$, for different values of $CV$. If the current fishing mortality is maintained, it leads to small values of equilibrium yield compared to the maximum sustainable yield; this situation becomes critical when there is high volatility in the recruitment (this is, for large values of $CV$), as shown in Figures \ref{fig:302}-\ref{fig:304}.  Indeed, Figure \ref{fig:302} emphasizes that any of the optimal constant fishing mortalities studied in this paper has a better performance in the equilibrium than the current fishing mortality for large levels of volatility. The same can be checked for the time evolution of SSB. In both cases the application of the current fishing mortality leads to a slow but constant decrease of the yield and SSB levels, concluding that overexploitation can lead to extinction, whereas the application of an optimal fishing mortality can maintain accepable levels of population even in scenarios with high volatility.
\begin{figure}[h!]
\begin{center}
\begin{minipage}{6cm}
\vspace{0.4cm}
\captionof{table}{Equilibrium values of yield and SSB (in tons) for the current fishing mortality $F_{\rm current}=0.292$, for different values of $CV$}\label{tab:aux001}
\vspace{0.4cm}
\begin{tabular}{||c||cc||}\hline\hline
  $CV$    & Mean Yield &  Mean  SSB \\ \hline \hline
    0     & 4223  & 17699 \\
  0.25    & 4192  & 17569 \\
  0.5     & 4094  & 17157 \\
  0.75    & 3890  & 16303 \\        
  1       & 3544  & 14852 \\
  1.5     & 2290  &  9603 \\
  2       &  827  &  3474 \\ \hline \hline
\end{tabular}
\end{minipage}
\ \ \hspace{0.5cm}
\begin{minipage}{8cm}
\includegraphics[scale=0.27]{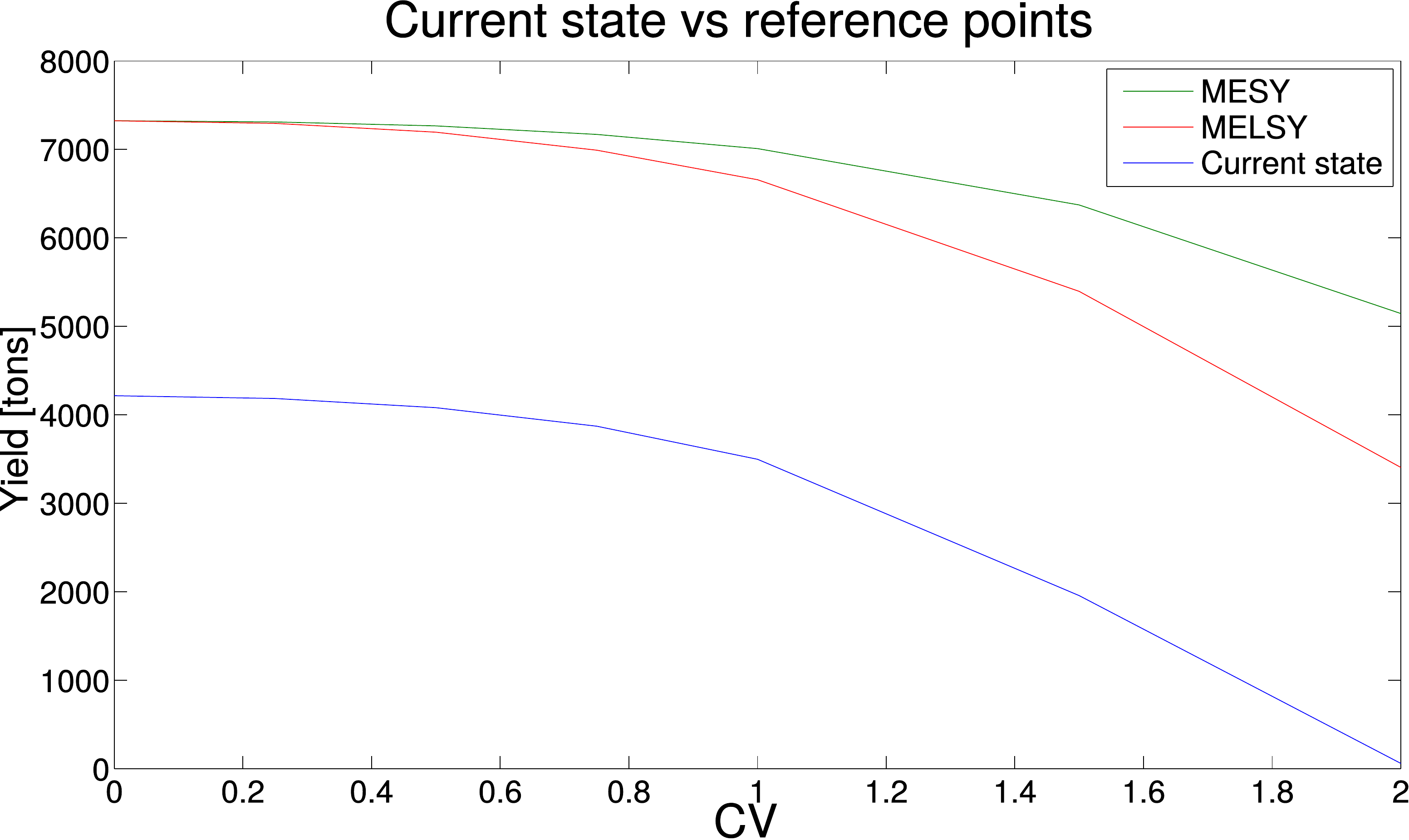}
\captionof{figure}{Yield (in tons) as function of $CV$ for the current fishing mortality $F_{\rm current}=0.292$, compared with the optimal yield measures}\label{fig:aux987}
\end{minipage}
\end{center}
\end{figure}

\begin{center}
\begin{minipage}{8cm}
\begin{center}
\includegraphics[scale=0.23]{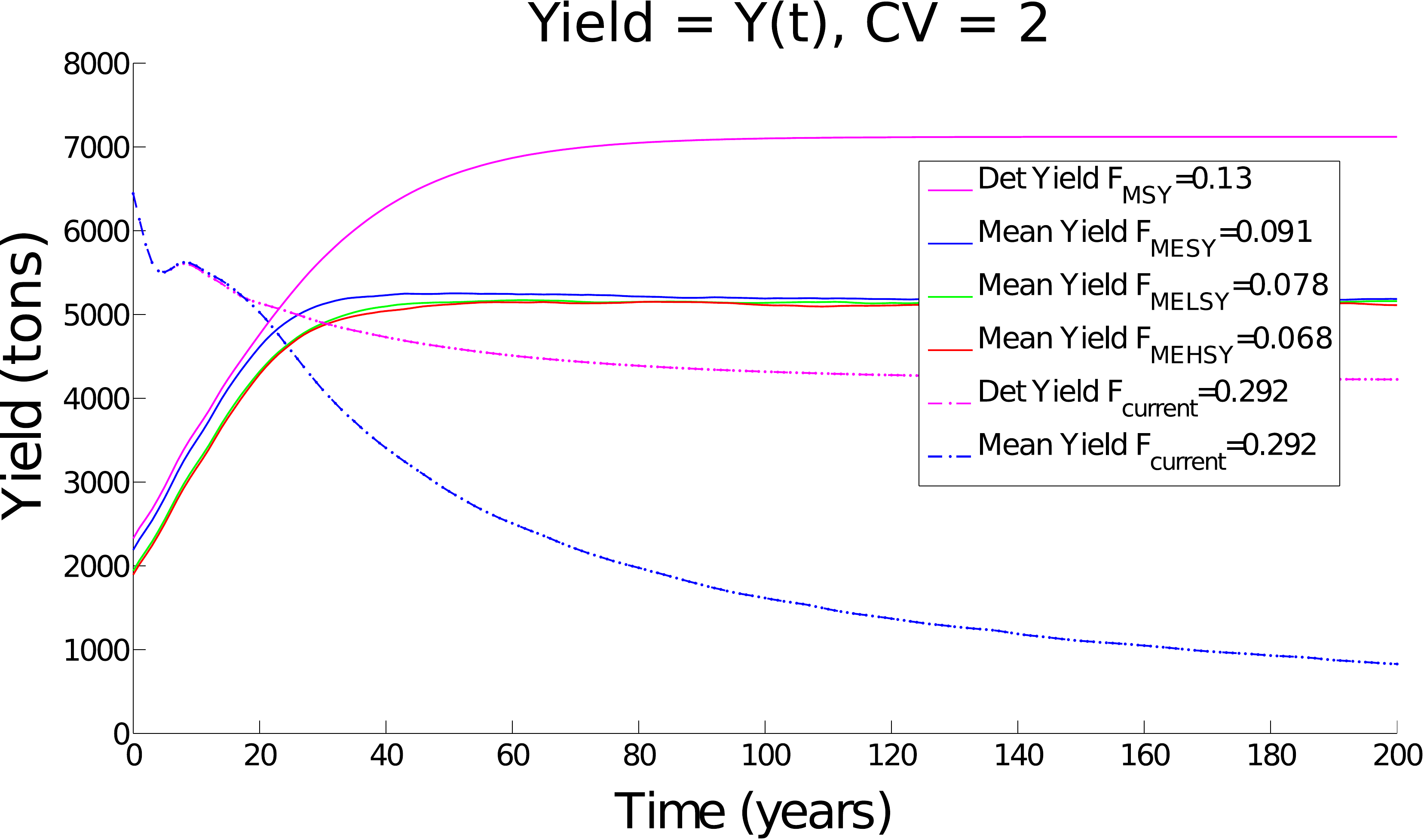}
\captionof{figure}{Yield measures for $CV=2$.}\label{fig:302}
\end{center}
\end{minipage}
%
%
%
\ \
\begin{minipage}{8cm}
\begin{center}
\includegraphics[scale=0.23]{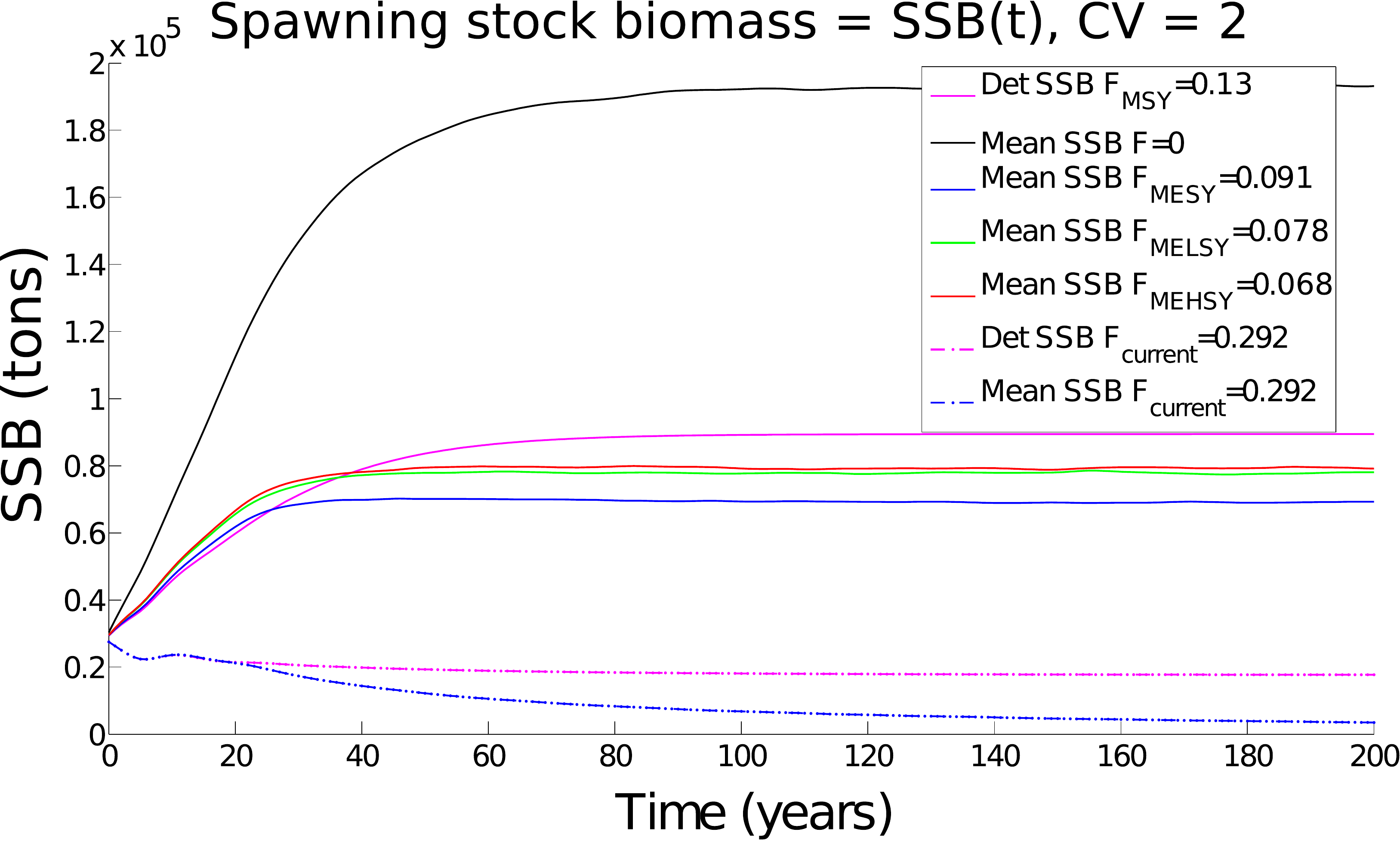}
\captionof{figure}{SSB for different yield measures, for $CV=2$.}\label{fig:304}
\end{center}
\end{minipage}
\end{center}

\section{Discussion}

Maximum sustainable yield is an important biological reference point that can be used to assess the status of fisheries and develop regulations as well as harvest control rules, and its estimation can be highly affected by the age structure of a fish population, and by multiple types of uncertainty, so biomass-based models and deterministic models are not accurate enough to compute the biological reference points associated to the respective fishery. In this work we propose to consider a mathematical model consisting in a stochastic single-species age-structured model for fisheries composed by several fleets that accounts for environmental and biological variability, and we extend the concept of MSY to the stochastic case via the concepts of MESY, MELSY, and MEHSY, that can be used as reference points provided that some degree of volatility is witnessed in the catches. We showed that this stochastic model could be applied to the study of the Patagonian toothfish fishery (Chilean and Argentinean stock), computing via numerical simulations using Monte Carlo method the maximum (constant) fishing mortalities $F_{\rm MESY}, F_{\rm MELSY}, F_{\rm MEHSY}$ for small and large volatility levels, and comparing them with the optimal deterministic fishing mortality $F_{\rm MSY}$ (theoretically obtained), concluding:
\begin{itemize}\itemsep0em
\item Yield is more variable and its mean value is smaller as the coefficient of variation increases. This fact needs to be considered when constructing confidence intervals for the maximum expected sustainable yields.
\item The maximum expected yields and optimal fishing mortalities decrease as the variability of fish recruitment increases, showing us the need to be more cautious when there is stochasticity present in recruitment.
\item For small volatility levels, the differences between MESY, MELSY, and MEHSY are not big, whereas for large levels of volatility they become apparent, in both the optimal mean values and the optimal fishing mortality. This is particularly important since fishing mortalities that allow a fishery to survive in the deterministic setting can be levels of overexploitation that lead to the extinction of the resource for high levels of volatility.
\end{itemize}

Moreover, the deterministic MSY is unlikely to be reached since volatility naturally occurs, and accurate methods to estimate this volatility are needed to asses the maximum expected yields and fishing mortalities. In this work we show that in order for sustainable harvest to occur, proper accounting of stochasticity in recruitment dynamics is mandatory. Conversely, not accounting for stochasticity in recruitment dynamics can lead in the worst case to extinction. Consequently, stochastic models should be used to develop regulation policies, especially for overexploited species in need of rebuilding. This is particularly important for decision makers at government management agencies.

\section{Acknowledgments}

This research was partially supported by FONDECYT project 1160204 and BASAL program CMM-AFB 170001 from CONICYT, Chile.


\bigskip



\appendix
\noindent{\bf\huge Appendix}
\medskip

\section{Adjustment of equilibrium equations for a plus group}\label{app:1}

Consider the equations for an age-structured dynamic with a plus group as in system \eqref{eq:modelBacalao}, written in short form as in \eqref{eq:matrix_form}. 
For the computation of the equilibrium we solve the equation 
\begin{equation}
N^{\star} = A(F)N^{\star} + B\varphi({\rm SSB}^{\star}),
\end{equation}
which translates to the equations
\begin{equation*}
\begin{split}
N_1^{\star} = \varphi({\rm SSB}^{\star}),\quad
N_{a+1}^{\star} = N_a^{\star}e^{-Z_a},\,\, a=1,\dots,A-2,\quad
N_{A}^{\star} = N_{A-1}^{\star}e^{-Z_{A-1}} + N_{A}^{\star}e^{-Z_{A}},
\end{split}
\end{equation*}
with 
\begin{equation}\label{eq:aux00}
{\rm SSB}^{\star} = \sum_{a=r}^{A}m_aW_aN_a^{\star}e^{-\tau Z_a}
\end{equation}
A recurrence formula for the abundances $N_{a+1}^{\star}$ can be derived, which depends on $N_1^{\star}$, for $a=2,\dots,A-1$:
\begin{equation}\label{eq:aux01}
N_a^{\star} = N_{a-1}^{\star}e^{-Z_{a-1}}=N_{a-2}^{\star}e^{-Z_{a-2}}e^{-Z_{a-1}} = \dots = N_1^{\star}\prod_{x=1}^{a-1}e^{-Z_{x}}.
\end{equation}
For the plus group, we have
\begin{equation}\label{eq:aux02}
N_{A}^{\star} = \frac{e^{-Z_{A-1}}}{1-e^{-Z_{A}}}N_{A-1}^{\star}= \frac{\prod_{x=1}^{A-1}e^{-Z_{x}}}{1-e^{-Z_{A}}}N_{1}^{\star}
\end{equation}
Define the \emph{cumulative survival} $\mathcal L_a$ and the \emph{spawning potential ratio} ${\rm SPR}^{\star}$ as
\begin{eqnarray}
\mathcal L_a &=&\prod_{x=1}^{a-1}e^{-Z_x} = \exp\left\{-\sum_{x=1}^{a-1}Z_x\right\},\quad\mathcal L_1=1,\label{eq:defLa}\\
{\rm SPR}^{\star} &=& \sum_{a=r}^{A-1}W_am_a\mathcal L_ae^{-\tau Z_a} + \frac{W_{A}m_{A}\mathcal L_{A}}{1-e^{-Z_{A}}}e^{-\tau Z_{A}}.\label{eq:defSPR}
\end{eqnarray}
Replacing \eqref{eq:aux01} and \eqref{eq:aux02} in \eqref{eq:aux00}, we obtain
\begin{equation}\label{eq:aux03}
{\rm SSB}^{\star} = N_1^{\star}{\rm SPR}^{\star}.
\end{equation}
where $N_1^{\star}$ solves the nonlinear equation:
\begin{equation}\label{eq:aux04}
N_1^{\star} = \varphi(N_1^{\star}{\rm SPR}^{\star}).
\end{equation}

Summarizing, the abundances at equilibrium solve
\begin{equation*}
\begin{split}
N_1^{\star} = \varphi(N_1^{\star}{\rm SPR}^{\star}),\quad N_{a}^{\star} = N_1^{\star}\mathcal L_a,\,\, a=2,\dots,A-1,\quad N_{A}^{\star} = N_{1}^{\star}\frac{\mathcal L_{A}}{1-e^{-Z_{A}}}.
\end{split}
\end{equation*}

With the previous values for $N^*$, the yield at equilibrium is function of the number of recruits $N_1^{\star}$
\begin{equation}\label{eq:yieldEquilibro}
Y^{\star} = \sum_{a=r}^{A} W_a \frac{F_a}{Z_a}(1-e^{-Z_a})N_{a}^{\star}= N_1^{\star}\left(\sum_{a=r}^{A-1} W_a \frac{F_a}{Z_a}(1-e^{-Z_a})\mathcal L_{a} + W_{A} \frac{F_{A}}{Z_{A}}\mathcal L_{A} \right).
\end{equation}

\begin{remark}
The spawning potential ratio ${\rm SPR^{\star}}$ represents the quantity of spawning stock biomass produced by one unit of recruits. 
\end{remark}

We can write ${\rm SPR^{\star}}$ in a simpler form. Defining 
\begin{equation*}
\tilde M_a = \sum_{x=1}^{a-1}M_x+\tau M_a,\quad\mbox{ and }\quad\tilde S_a = \sum_{f=1}^{n}P_f\left(\sum_{x=1}^{a-1}s_{f,x} + \tau s_{f,a}\right),
\end{equation*}
we have
\begin{equation}\label{eq:SPRsimple}
{\rm SPR}^{\star}(F) = \sum_{a=r}^{A-1}W_am_ae^{-\tilde M_a}e^{-\tilde S_a F} + W_{A}m_{A}\frac{ e^{-\tilde M_{A}}e^{-\tilde S_{A}F}}{1-e^{-M_{A}}e^{-\sum_{f=1}^nP_fs_{f,A}F} }.
\end{equation}

\begin{proposition}\label{prop:SPR}
The function $F\mapsto {\rm SPR^{\star}(F)}$ is decreasing and converges to 0 as $F$ goes to infinity.
\end{proposition}

\begin{proof}
The derivatives of $Z_a$ and $\mathcal L_a$ with respect to $F$ are
\begin{equation}\label{aux:0001}
\frac{\partial Z_a(F)}{\partial F} = \sum_{f=1}^ns_{f,a}P_f,\qquad\frac{\partial \mathcal L_a(F)}{\partial F} = -\mathcal L_a\sum_{f=1}^n\sum_{x=1}^{a-1}s_{f,x}P_f.
\end{equation}
We compute the derivative of ${\rm SPR^{\star}}$ (given in \eqref{eq:SPRsimple}) with respect to $F$, using \eqref{aux:0001}:
\begin{equation*}
\begin{split}
\frac{\partial{\rm SPR^{\star}}(F)}{\partial F} &= -\sum_{a=r}^{A-1} W_am_a\mathcal L_ae^{-\tau Z_a}\left[ \sum_{f=1}^nP_f\left(ps_{f,a} + \sum_{x=1}^{a-1}s_{f,x} \right) \right]\\
& \quad- \frac{W_{A}m_{A}}{(1-e^{-Z_{A}})^2}\mathcal L_{A}e^{-pZ_{A}}\left\{ \sum_{f=1}^nP_f\left[\left( ps_{f,{A}}+\sum_{x=1}^{{A}-1}s_{f,x} \right)(1-e^{-Z_{A}}) + e^{-Z_{A}}s_{f,{A}} \right] \right\},
\end{split}
\end{equation*}
which is negative for all values of $F$, thus proving that $F\mapsto{\rm SPR}^{\star}(F)$ is decreasing. For the second statement, we take limits in \eqref{eq:SPRsimple} when $F\rightarrow\infty$ and we see that both terms of the right-hand side converge to zero.\qed
\end{proof}

The equilibrium spawning potential ratio ${\rm SPR}^{\star}$ is maximized in $F=0$:
\begin{equation}\label{eq:SPRmax}
{\rm SPR}^{\star}(0) = \sum_{a=r}^{A-1}W_am_ae^{-\tilde M_a} + W_{A}m_{A}\frac{ e^{-\tilde M_{A}}}{1-e^{-M_{A}} }.
\end{equation}
This value is important for the existence of a positive equilibrium value of the equilibrium abundances. For the particular case of a Beverton-Holt spawner-recruit function of the form
\begin{equation}\label{eq:bvfunction2}
\varphi({\rm SSB}) = \frac{\alpha {\rm SSB}}{\beta + {\rm SSB}}, 
\end{equation}
the equilibrium $\rm SSB^{\star}$ and abundance $N_1^{\star}$ are
\begin{equation*}
{\rm SSB}^{\star} = \alpha{\rm SPR}^{\star}-\beta,\qquad N_1^{\star} = \alpha-\frac{\beta}{{\rm SPR^{\star}}}.
\end{equation*}
Then, a condition for the existence of $F\geq0$ such that the corresponding equilibrium point is positive is that 
\begin{equation}\label{eq:condicionExistencia}
{\rm SPR}^{\star}(0) = \sum_{a=r}^{A-1}W_am_ae^{-\tilde M_a} + W_{A}m_{A}\frac{ e^{-\tilde M_{A}}}{1-e^{-M_{A}} } > \frac{\beta}{\alpha}.
\end{equation}
This condition is completely related to the capacity of the fish population to survive in the environment. We conclude the following proposition:

\begin{proposition}\label{prop:existencia}
Suppose that the spawner-recruit function is of the form \eqref{eq:bvfunction2} and that condition \eqref{eq:condicionExistencia} is satisfied. Then, there exists a value $F_{\rm ext}>0$ that solves the equation
\begin{equation}\label{eq:conclusion_existence}
{\rm SPR^{\star}}(F_{\rm ext})= \frac{\beta}{\alpha},
\end{equation}
and then $F^{\star}_{\rm MSY}$ belongs to the interval $[0,F_{\rm ext}]$.
\end{proposition}

\begin{proof}
Consider the function $g(F):={\rm SPR^{\star}}(F)-\frac{\beta}{\alpha}$. This is a continuous and differentiable function, since $F\mapsto {\rm SPR}^{\star}(F)$ is continuous and differentiable. Thanks to \eqref{eq:condicionExistencia} we have $g(0)>0$; from Proposition \ref{prop:SPR} we see that $F\mapsto g(F)$ is decreasing and $\lim_{F\rightarrow\infty} g(F)=-\frac{\beta}{\alpha}<0$. By the intermediate value theorem there exists then a point $F_{\rm ext}$ such that $g(F_{\rm ext})=0$, or equivalently, \eqref{eq:conclusion_existence} is satisfied. The fact that $F^{\star}_{\rm MSY}\in [0,F_{\rm ext}]$ is concluded because for any $F>F_{\rm ext}$ the value of equilibrium of ${\rm SSB^{\star}(F)}$ and $N_1^{\star}(F)$ are negative, leading to negative equilibrium yield from the formula \eqref{eq:yieldEquilibro}.\qed
\end{proof}

\section{Parameters}\label{app:2}

In this section we present the values of the parameters that are too long to be detailed in the main body of the paper. The selectivity by age and fleet is given in Table \ref{table:app1}. In Table \ref{table:app2} we show the initial condition for the model. Since the recruitment at a given year depends of the spawning stock biomass of the previous year, we should take as initial condition the data corresponding to the last year of the measurements. Here, $N_{-2}$ shows the abundance by age (in millions) for the year 2012, $N_{-1}$ corresponds to year 2013 and $N_{0}$ corresponds to year 2014, and $N_{-2}$ and $N_{-1}$ are used to estimate the initial abundances for ages 1 and 2 for $N_0$. In Table \ref{table:app3} we show the values of weight, maturity, and natural mortality considered for the numerical simulations.
\begin{center}
\captionof{table}{Selectivity by age and fleet from Tascheri and Canales \cite{bacalao}.}\label{table:app1}
\begin{tabular}{||c||c|c|c|c||c||c|c|c|c||}\hline\hline
Age (a) & Fleet 1 & Fleet 2 & Fleet 3 & Fleet 4  & Age (a) & Fleet 1 & Fleet 2 & Fleet 3 & Fleet 4 \\ \hline \hline
1   &       0  &       0  &       0  &       0  & 16  &  0.6347  &  0.3589  &  1.0000  &  1.0000  \\
2   &       0  &       0  &       0  &       0  & 17  &  1.0000  &  0.3589  &  1.0000  &  1.0000  \\
3   &  0.0002  &  0.0032  &  0.0008  &  0.0379  & 18  &  1.0000  &  0.3589  &  1.0000  &  1.0000  \\
4   &  0.0008  &  0.0113  &  0.0025  &  0.0600  & 19  &  1.0000  &  0.3589  &  1.0000  &  1.0000  \\
5   &  0.0032  &  0.0427  &  0.0076  &  0.0946  & 20  &  1.0000  &  0.3589  &  1.0000  &  1.0000  \\
6   &  0.0128  &  0.1969  &  0.0227  &  0.1493  & 21  &  1.0000  &  0.3589  &  1.0000  &  1.0000  \\
7   &  0.0468  &  0.7109  &  0.0658  &  0.2372  & 22  &  1.0000  &  0.3589  &  1.0000  &  1.0000  \\
8   &  0.1051  &  1.0000  &  0.1829  &  0.3766  & 23  &  1.0000  &  0.3589  &  1.0000  &  1.0000  \\
9   &  0.1491  &  0.8079  &  0.4390  &  0.5757  & 24  &  1.0000  &  0.3589  &  1.0000  &  1.0000  \\
10  &  0.1776  &  0.5793  &  0.7773  &  0.7980  & 25  &  1.0000  &  0.3589  &  1.0000  &  1.0000  \\
11  &  0.2103  &  0.4176  &  0.9813  &  0.9562  & 26  &  1.0000  &  0.3589  &  1.0000  &  1.0000  \\
12  &  0.2904  &  0.3589  &  1.0000  &  1.0000  & 27  &  1.0000  &  0.3589  &  1.0000  &  1.0000  \\
13  &  0.3689  &  0.3589  &  1.0000  &  1.0000  & 28  &  1.0000  &  0.3589  &  1.0000  &  1.0000  \\
14  &  0.4262  &  0.3589  &  1.0000  &  1.0000  & 29  &  1.0000  &  0.3589  &  1.0000  &  1.0000  \\
15  &  0.4822  &  0.3589  &  1.0000  &  1.0000  & 30$^{+}$ &  1.0000  &  0.3589  &  1.0000  &  1.0000  \\ \hline\hline
\end{tabular}
\end{center}

\begin{center}
\captionof{table}{Initial abundances (in millions) from Tascheri and Canales \cite{bacalao}. 
}\label{table:app2}
\begin{tabular}{||c||c|c|c||c||c|c|c||}\hline\hline
Age $(a)$ & $N_{-2}$ & $N_{-1} $ & $N_{0}$  & Age $(a)$ & $N_{-2}$ & $N_{-1}$ & $N_{0}$ \\ \hline\hline
1  &       *  &       *  &  2742.5 (est) & 16  &  131.8  &   138.6  &   171.9   \\
2  &       *  &       *  &  2847.4 (est) & 17  &  111.4  &    92.1  &    97.3   \\
3  &  1889.4  &  2492.1  &  2993.9  & 18  &   71.5  &    72.3  &    59.0   \\
4  &  1314.6  &  1623.5  &  2141.4  & 19  &   51.1  &    46.4  &    46.3   \\
5  &  1407.7  &  1127.8  &  1392.9  & 20  &   30.1  &    33.2  &    29.7   \\
6  &  1467.0  &  1202.5  &   963.9  & 21  &   25.6  &    19.5  &    21.3   \\
7  &  1499.2  &  1231.7  &  1012.5  & 22  &   18.9  &    16.6  &    12.5   \\
8  &  1189.8  &  1188.9  &   985.8  & 23  &   17.5  &    12.3  &    10.6  \\
9  &   735.4  &   902.3  &   910.5  & 24  &   11.3  &    11.4  &     7.9  \\
10  &  516.7  &   554.3  &   685.7  & 25  &    8.9  &     7.3  &     7.3  \\
11  &  541.1  &   386.7  &   416.4  & 26  &    5.8  &     5.8  &     4.7  \\
12  &  477.0  &   404.0  &   287.5  & 27  &    4.6  &     3.8  &     3.7  \\
13  &  429.4  &   354.7  &   298.7  & 28  &    2.9  &     3.0  &     2.4  \\
14  &  324.5  &   317.3  &   261.1  & 29  &    2.1  &     1.9  &     1.9  \\
15  &  192.5  &   237.4  &   232.2  & 30$^{+}$ &   10.5  &     8.2  &     6.5  \\ \hline\hline
\end{tabular}
\end{center}

\begin{center}
\captionof{table}{Weight, maturity, and natural mortality by age, from Tascheri and Canales \cite{bacalao}.}\label{table:app3}
\begin{tabular}{||c||c|c|c||c||c|c|c||}\hline\hline
Age (a) & Weight & Maturity & Natural mortality   & Age (a) & Weight & Maturity & Natural mortality  \\ \hline \hline
1   &       0  &    0   &       0  &        16  &  17.5183  &  0.99   &  0.15  \\
2   &       0  &    0   &       0  &        17  &  18.8940  &  0.99   &  0.15  \\
3   &  1.1123  &  0.01   &  0.15  &   18  &  20.5717  &  1     &  0.15  \\
4   &  1.2620  &  0.02   &  0.15  &   19  &  21.8620  &  1     &  0.15  \\
5   &  1.5820  &  0.03   &  0.15  &   20  &  24.6300  &  1     &  0.15  \\
6   &  2.4400  &  0.06   &  0.15  &   21  &  26.1187  &  1     &  0.15  \\
7   &  3.4893  &  0.11   &  0.15  &   22  &  28.0610  &  1     &  0.15  \\
8   &  4.8843  &  0.21   &  0.15  &   23  &  29.1047  &  1     &  0.15  \\
9   &  6.5013  &  0.34   &  0.15  &   24  &  31.1807  &  1     &  0.15  \\
10  &  8.4117  &  0.51   &  0.15  &   25  &  33.3417  &  1     &  0.15  \\
11  &  10.0500  &  0.68   &  0.15  &   26  &  35.9160  &  1     &  0.15  \\
12  &  11.3690  &  0.81   &  0.15  &   27  &  35.9857  &  1     &  0.15  \\
13  &  12.8947  &  0.90   &  0.15  &   28  &  43.1203  &  1     &  0.15  \\
14  &  14.6080  &  0.95   &  0.15  &   29  &  43.4607  &  1     &  0.15  \\
15  &  16.1203  &  0.97   &  0.15  &   30$^{+}$ &  51.8067  &  1   &  0.15    \\ \hline\hline
\end{tabular}
\end{center}

\end{document}